\documentclass[referee,sn-basic]{sn-jnl}% referee option is meant for 

\usepackage{subcaption}
\usepackage{titleps}
\usepackage{url}
\usepackage{rotating}
\usepackage{float}
\usepackage{graphicx}
\usepackage{booktabs}
\usepackage{multirow}
\usepackage{caption}
\usepackage{amssymb}
\usepackage{amsmath}
\usepackage{mathtools}
\usepackage{algorithm}
\usepackage{algpseudocode}
\usepackage{comment}
\usepackage{bbm}
\usepackage{relsize}
\DeclareMathOperator*{\argmax}{arg\,max}
\DeclareMathOperator*{\argmin}{arg\,min}
\DeclareMathOperator*{\maxlen}{\operatorname{MAXLEN}}

\newcommand{\defeq}{\vcentcolon=}

\newtheorem{prop}{Proposition}
\begin{document}
\title{Outlier detection for mixed-type data: A novel approach}


\author*[1]{\fnm{Efthymios} \sur{Costa}}\email{efthymios.costa17@imperial.ac.uk}

\author[1]{\fnm{Ioanna} \sur{Papatsouma}}\email{i.papatsouma@imperial.ac.uk}

\affil*[1]{\orgdiv{Department of Mathematics}, \orgname{Imperial College London}} %\orgaddress{\street{Street}, \city{City}, \postcode{100190}, \state{State}, \country{Country}}}
%\orgaddress{\street{Nea Hili}, \city{Alexandroupoli}, \postcode{68100}, \country{Greece}}}


\abstract{Outlier detection can serve as an extremely important tool for researchers from a wide range of fields. From the sectors of banking and marketing to the social sciences and healthcare sectors, outlier detection techniques are very useful for identifying subjects that exhibit different and sometimes peculiar behaviours. When the data set available to the researcher consists of both discrete and continuous variables, outlier detection presents unprecedented challenges. In this paper we propose a novel method that detects outlying observations in settings of mixed-type data, while reducing the required user interaction and providing general guidelines for selecting suitable hyperparameter values. The methodology developed is being assessed through a series of simulations on data sets with varying characteristics and achieves very good performance levels. Our method demonstrates a high capacity for detecting the majority of outliers while minimising the number of falsely detected non-outlying observations. The ideas and techniques outlined in the paper can be used either as a pre-processing step or in tandem with other data mining and machine learning algorithms for developing novel approaches to challenging research problems.}

%189 words, max 200 words

\maketitle

\section{Introduction}

Outlier detection aims to flag atypical observations in a data set; these are the `outliers' (also called anomalies) and they may be points with values that do not conform to some pattern suggested by the data or observations whose values arouse suspicion regarding the mechanism that has been used to generate them. The process of finding outliers within a data set is usually not straightforward and requires special care, as these observations may be indications of a serious offense or a critical situation. For instance, they could be intrusions in a network \citep{aggarwal2007data, di2008intrusion}, cases of financial fraud or money laundering \citep{ngai2011application}, malicious individuals in online social networks \citep{savage2014anomaly} or pathologies in medical images \citep{tschuchnig2022anomaly}, among others. The wide range of uses of outlier detection (also known as anomaly detection or novelty detection) in a variety of domains, such as the aforementioned ones, has driven the development of numerous algorithms designed to detect anomalies within a given data set. However, despite the increasing popularity of many of these techniques, most of them are restricted to just one type of data and that one is mainly continuous data.\\

Mixed-type data, that is data consisting of both continuous and discrete (also known as categorical) variables, is commonly encountered in plenty of fields. As an example, clinical records may include information such as a patient's age, weight and height, as well as demographic characteristics like gender, race and marital status. In marketing research, discrete variables that portray a client's demographic, psychographic or socio-economic background, together with quantitative data related to their purchase behaviour are used for advertising products or services tailored to their needs and interests. The list of examples of mixed-type data uses in real-world applications is non-exhaustive, underscoring the importance of developing efficient methods capable of detecting data abnormalities in mixed-attribute domains.

To the best of our knowledge, the first algorithm that was developed to detect anomalies in mixed-attribute data sets was the Link-based Outlier and Anomaly Detection in Evolving Data Sets (LOADED) algorithm of \cite{ghoting2004loaded}, who suggested the use of association rules from the frequent itemset mining literature \citep{agrawal1994fast} for detecting anomalies in the discrete space. The core limitation of the LOADED algorithm is the use of covariance matrices to detect outliers in the continuous domain, something which requires large computational memory; moreover, the method does not look for interactions between discrete and continuous variables, thus failing to detect anomalies in the mixed-attribute space successfully. An improvement of this algorithm in terms of computational cost was then presented by \cite{otey2006fast}, with further ameliorations in the detection of continuous outliers being proposed by \cite{koufakou2010}, who introduced the Outlier Detection for Mixed Attribute Datasets (ODMAD) algorithm. The primary concern with ODMAD is that it does not account for interactions between discrete and continuous features, besides making use of the cosine similarity. However, this does not guarantee that the majority of outliers in the mixed-attribute space will be detected and in fact, clustered anomalies are very likely to be missed by ODMAD. Furthermore, the ODMAD algorithm is not fully unsupervised as its implementation involves certain threshold parameters, the values of which need to be defined by the user. Different threshold values can lead to very different results and the authors tune these based on existing labels that indicate which observations are outlying in the data set.

A different approach was taken by \cite{zhang2011effective}, who developed a Pattern-based Outlier Detection (POD) algorithm that uses logistic regression to assign scores of outlyingness to observations of mixed-type. POD flags the top $k$ points with the highest scores as outliers, where the value of $k$ is chosen according to the user's belief on the proportion of outliers included in the data set, hence making it hard to specify the value of $k$ in a completely unsupervised setting. The work of \cite{bouguessa2015practical} attempts to overcome these issues by looking at the problem from a different point of view. More precisely, anomaly detection is seen from a mixture modelling perspective, with outlier scores being calculated independently for discrete and continuous attributes and a bivariate Beta mixture model being fit on these two score variables (upon normalisation). Anomalies are then defined as the observations that have been assigned to the mixture component with the highest average discrete and continuous scores of outlyingness. Despite the fact that this approach does not rely on user-defined threshold values for the scores, it is limited to calculating scores separately for different data types, just like what is being done by the ODMAD algorithm. However, outliers in a mixed-attribute space may as well have low scores of outlyingness in both the discrete and the continuous domains, yet they may still be atypical.

Recent advances in the field of deep learning have led to many state-of-the-art deep learning techniques gaining increasing attention and also being used for detecting outliers. The use of deep learning methods for anomaly detection is sometimes being referred to as Deep Anomaly Detection and some of the most popular such algorithms include Generative Adversarial Networks (GAN) \citep{schlegl2017unsupervised}, autoencoder networks \citep{an2015variational, chen2018autoencoder} and Restricted Boltzmann Machines (RBMs) \citep{fiore2013network}; a complete overview of deep learning techniques for anomaly detection can be found in \cite{pang2021deep}. One of the main drawbacks that these methods present is that they can only handle mixed-type data using one-hot encoding for discrete variables which usually leads to a sparse representation of the data and which can become very challenging to deal with, as described in \cite{thudumu2020comprehensive}. Recently, \cite{do2016outlier} proposed an energy-based approach for outlier detection of mixed-type data, extending the work on Mixed-variate RBMs (Mv.RBM) of \cite{tran2011mixed}. However, the number of layers to be used, as well as the model architecture of an autoencoder built for anomaly detection are not always trivial to the user, thus posing a serious limitation.

In this paper we propose a novel method for detecting outliers in mixed-attribute space. The main contributions of this work are as follows: we improve upon existing approaches to the problem of outlier detection in mixed-attribute domain, leveraging key insights that stem from prior research in the field and further reducing the amount of required user input throughout the process. We distinguish among several types of outliers that a data set may include by thoroughly defining these in Section \ref{sec:defoutliers}. We then define scores of outlyingness which are used for quantifying the likelihood of a data point being anomalous in Section \ref{sec:scores}. Sections \ref{sec:detectmargouts} and \ref{sec:detectjointouts} focus on the main tools and techniques used for successfully flagging outlying observations of different types in a mixed feature space, while limiting the number of falsely detected non-outliers. The efficacy of our method is illustrated via a large number of simulations on data sets with varying characteristics.

\section{Definition of outliers}\label{sec:defoutliers}

The definition of an outlier is crucial in the development of a technique that can detect data abnormalities. Especially in the case of mixed-type data, we typically face the problem of not having a proper definition of outliers. Our intuition is that anomalies may appear in the discrete or the continuous space, but they could also exist in both domains. Moreover, anomalies can be detected in the mixed-attribute space as well, meaning that they may not be outlying in either domain, yet they may be anomalous because they do not conform to an existing (and usually unknown) pattern between the discrete and the continuous features. In this section, we define these types of outliers and explain the strategy that we follow in order to detect them. We have also summarised some useful notation in Table \ref{tab:notation}; notice that the notation appears in the order with which each term appears in the remainder of the paper.
\vspace{-0.15cm}
\begin{table}[h!]
\centering
\footnotesize{
\begin{tabular}{ll}
Notation & Definition           \\[0.025cm] \hline
$\boldsymbol{X}$ & Data set including both discrete and continuous features\\[0.025cm]
$\boldsymbol{X}_{i,D}$ & Discrete feature values of the $i$th observation\\[0.025cm]
$\boldsymbol{X}_{i,C}$ & Continuous feature values of the $i$th observation\\[0.025cm]
$\boldsymbol{X}_{D_j}$ & Discrete variable $j$\\[0.025cm]
$\boldsymbol{X}_{C_j}$ & Continuous variable $j$\\[0.025cm]
$n$ & Number of observations\\[0.025cm]
$p$ & Total number of features\\[0.025cm]
$p_D$ & Number of discrete features\\[0.025cm]
$p_C$ & Number of continuous features\\[0.025cm]
$\ell_j$ & Number of levels of $j$th discrete variable\\[0.025cm]
$\mathcal{O}$ & Set of outliers\\[0.025cm]
$\mathcal{O}_M$ & Set of marginal outliers\\[0.025cm]
$\mathcal{O}_J$ & Set of joint outliers\\[0.025cm]
$\mathcal{E}$ & Set of unique discrete score values\\[0.025cm]
$\rho$ & Maximum proportion of outliers believed to be in the data set\\[0.025cm]
$\epsilon$ & Additional proportion of outliers that we are willing to tolerate\\[0.025cm]
$s_{D,(i, \cdot)}$ & Discrete score of outlyingness for observation $\boldsymbol{X}_i$\\[0.025cm]
$c_{D,(i, j)}$ & Contribution of $j$th discrete variable to $\mathcal{S}_{D,(i, \cdot)}$\\[0.025cm]
$s_{C,(i, \cdot)}$ & Continuous score of outlyingness for observation $\boldsymbol{X}_i$\\[0.025cm]
$d$ & Sequence of discrete levels (itemset)\\[0.025cm]
$\lvert d \rvert$ & Length of itemset $d$\\[0.025cm]
$\operatorname{supp}(d)$ & Support/Frequency of itemset $d$\\[0.025cm]
$\maxlen$ & Maximum length of itemset $d$\\[0.025cm]
$\sigma_d$ & Minimum frequency threshold for itemset $d$\\[0.025cm]
$u_{j,j'}$ & Theil's $U$ score between discrete variables $j$ and $j'$\\[0.075cm]
$u_{j,j'}^\mathrm{upper}$ & Upper threshold value for $u_{j,j'}$\\[0.025cm]
$m_D$ & Mean of discrete scores $s_{D,(i,\cdot)}$\\[0.025cm]
$m_C$ & Mean of sorted differences between continuous scores $s_{C,(i,\cdot)}$\\[0.025cm]
$s_D$ & Standard error of discrete scores $s_{D,(i,\cdot)}$\\[0.025cm]
$s_C$ & Standard error of sorted differences between continuous scores $s_{C,(i,\cdot)}$\\[0.025cm]
$\mathcal{I}_D$ & Set of indices of marginal outliers in discrete space\\[0.025cm]
$\mathcal{I}_C$ & Set of indices of marginal outliers in continuous space\\[0.025cm]
$\mathcal{J}$ & Set of indices of the continuous variables associated with $\boldsymbol{X}_{D_j}$ \\[0.025cm]
$\delta$ & Proportion of nearest neighbours\\[0.025cm]
$n_{j,l}$ & Number of non-marginally outlying observations for which $\boldsymbol{X}_{D_j} = l$\\[0.025cm]
$\pi_{j,l}$ & Proportion of non-marginally outlying observations for which $\boldsymbol{X}_{D_j} = l$\\[0.025cm]
$\alpha_1$ & Significance level of Kruskal-Wallis $H$ test\\[0.025cm]
$\alpha_2$ & Significance level of chi-square goodness of fit test\\[0.025cm]
%$\hat{\boldsymbol{Y}}$ & $\left(n \times r \right)$-dimensional matrix of predicted probabilities for $n$ observations and $r$ possible outcomes\\[0.025cm]
%$w_{i,j}$ & Weight of observation $i$ when trying to predict variable $j$\\[0.025cm]
%$\mathcal{L}^W$ & Weighted cross-entropy loss function\\[0.025cm]
$\hat{f}_l$ & Kernel Density Estimator of the density of the $l$th level of a discrete variable\\[0.025cm]
$\Lambda_i$ & \begin{tabular}{@{}l@{}}Ratio of maximum KDE value to that of the KDE of the true level of \\the $i$th observation\end{tabular}\\[0.025cm]
$\Lambda^*$ & KDE ratio threshold for detecting joint outliers\\[0.025cm]
$\Lambda_m^*$ & \begin{tabular}{@{}l@{}}Small KDE ratio threshold value to be used should too many points \\are misclassified using method of consecutive angles\end{tabular}\\[0.025cm]
$N\left(\Lambda^*\right)$ & Number of misclassified points using KDE for which $\Lambda_i > \Lambda^*$\\[0.025cm]
$\theta_{\Lambda^*}$ & \begin{tabular}{@{}l@{}}Angle between the horizontal axis and the line segment connecting\\ $N\left(\Lambda^*-0.5\right)$ and $N\left(\Lambda^*\right)$\end{tabular}\\[0.025cm]
$\Lambda^*_\mathrm{elbow}$ & \begin{tabular}{@{}l@{}}Value of $\Lambda^*$ where the elbow of the curve of misclassified observations\\ by KDE against $\Lambda^*$ is located \end{tabular}\\[0.025cm]
$\theta_\mathrm{elbow}$ & Angle between the segments joining the elbow point with $N(1)$ and $N(20)$\\[0.025cm]
$\theta_\mathrm{thresh}$ & \begin{tabular}{@{}l@{}}Threshold value for $\theta_\mathrm{elbow}$, determining whether $\Lambda^*=3$ or the\\ method of consecutive angles should be used\end{tabular}
\end{tabular}}
\caption{Useful notation; terms are mentioned in the order that they appear within the rest of the manuscript.}
\label{tab:notation}
\end{table}

\subsection{Marginal outliers}

The first type of outliers that we define is the observations which are found to be outlying when we consider discrete and continuous features separately. These are referred to as `marginal outliers'. Marginal outliers that are anomalous in either just the discrete or just the continuous space are called `single marginal' outliers, while if an observation is outlying in both domains separately, it is referred to as a `combined marginal' outlier. In order to fix some notation, we assume that our data set $\boldsymbol{X}$ consists of a combination of discrete and continuous variables and includes $n$ observations. We denote the set of outliers by $\mathcal{O}$, while the set of marginal outliers is denoted by $\mathcal{O}_M$. We further introduce $\mathcal{O}_D$ and $\mathcal{O}_C$ as the sets of observations which are marginally outlying in just the discrete or just the continuous domains, respectively. Therefore, we say that a data point $\boldsymbol{X}_i$ ($i=1, \dots, n$) is a single marginal outlier if and only if $\boldsymbol{X}_i \in \left(\mathcal{O}_D \cap \overline{\mathcal{O}}_C\right) \cup \left(\overline{\mathcal{O}}_D \cap \mathcal{O}_C \right)$, where $\overline{A}$ denotes the complement of set $A$. The set of combined marginal outliers is $\mathcal{O}_D \cap \mathcal{O}_C$ and we can express the set of marginal outliers as $\mathcal{O}_M = \mathcal{O}_D \cup \mathcal{O}_C$.

We define outliers in the discrete space in the same way as \cite{otey2006fast}, who use ideas from the frequent itemset mining literature to introduce links between variables in the categorical space. More precisely, they claim that an outlier in the discrete domain is an observation that either includes a categorical level that does not appear frequently within a discrete variable or it may as well be a data point that includes a rare co-occurrence of levels of two or more categorical features. This definition of discrete outliers was shown to be very effective in practice in \cite{koufakou2007scalable}, with more details on which discrete levels are considered as frequent being given in Section \ref{sec:scores}. Outliers in the continuous domain are defined as observations that are far away from the rest of the data and which appear to be isolated in the $p_C$-dimensional space of continuous variables (where $p_C$ is the number of continuous variables in the data set). The advantage of this definition of continuous outliers is that it enables us to detect points in regions of extremely low density and flag them as outliers, without having to make any distributional assumptions about the data generating process.

\subsection{Joint outliers}

While the marginal outliers are observations which deviate much from the rest of the data, there may as well be anomalies that look perfectly normal at first sight. We define these observations to be the `joint outliers', which only exist in the joint space of discrete and continuous variables and are therefore not marginally outlying in either domain. Our motivation is that in many practical applications, variables of different types may be found to be associated in a certain way that deems any observations violating such relationships as outliers. These relationships may be known to the user, although in most practical applications, this will not be the case. As a result, it is important that the method we use to detect outliers in a mixed data set can disentangle any interactions among discrete and continuous features so that abnormal observations in the mixed space can be uncovered. Denoting the set of joint outliers by $\mathcal{O}_J$ and the set of marginal outliers by $\mathcal{O}_M$, we recall that these are completely disjoint by construction. Hence, we can partition the set of outliers $\mathcal{O}$ as the union of the sets of marginal and joint outliers; $\mathcal{O} = \mathcal{O}_M \cup \mathcal{O}_J$.

In order to motivate the importance of detecting the joint outliers in mixed-type data, we consider the scenario of a data set including clinical records of patients in a hospital. Suppose that different types of treatment are given to patients suffering from a disease based on their body mass index (BMI). The BMI is calculated as the ratio of the weight (in kilograms) to the square of the height (in meters) of a patient and different BMI value ranges correspond to different classifications (underweight, normal, overweight, obese). If the only data available is the height and the weight of each patient, as well as the type of treatment given to them, it is very unlikely that the user can detect instances of the wrong treatment being given to a patient. Had the user known that the treatment type is determined by the patient's BMI, which is a function of the two continuous variables available to them, they could have easily run a quick check and detect such data irregularities, but having access to information from the data collection team is seldom possible.

\section{Scores of outlyingness}\label{sec:scores}

We present the scores of outlyingness that we calculate for data points in both the discrete and the continuous domains, which are going to be used for detecting marginal outliers. Notice that these scores can vary for different data sets, meaning that there is no universal threshold value for the score that can be used to determine whether an observation is an outlier. 

\subsection{Discrete Score}

We recall our definition of marginal outliers in the discrete space; these are observations which include discrete levels that are infrequent within their respective discrete variables, or data points which include rare co-occurrences of levels or sequences of categorical attributes. Based on ideas taken from the association rule mining literature \citep{agrawal1994fast}, we present a method for calculating scores of outlyingness for the discrete features of each observation, which we will be referring to as `discrete scores'. The method can be seen as an improved and automated version of the discrete score that was used by \cite{koufakou2010} in the development of the ODMAD algorithm, with suitable threshold values being determined in a data-driven manner.

In order to introduce the method used to compute the discrete scores, we first make a brief introduction to the concept of an itemset. An itemset $d$ is a set consisting of one discrete level or of a sequence of discrete levels. The number of times that an itemset $d$ occurs within a data set is called the support of $d$ and it is denoted by $\operatorname{supp}(d)$. The itemset length, denoted by $\lvert d \rvert$ is defined as the number of categorical variables that have been considered to produce itemset $d$. Assuming that our data set includes $p_D$ discrete variables, we have a total of $2^{p_D}-1$ distinct sequences of variables (where a sequence of unit length refers to just the discrete variable itself) and if each discrete variable $j$ includes $\ell_j$ discrete levels, the total number of itemsets that one could encounter is equal to $\prod_{j=1}^{p_D}\left(\ell_j+1 \right)-1$.

Based on the definition of discrete outliers, we can look for itemsets with low support in order to find outlying observations in the categorical space. We define the discrete score for an observation $i$ as:

\begin{equation}\label{eq:discretescorefinal}
    s_{D,(i,\cdot)}=\sum_{\substack{d \subseteq \boldsymbol{X}_{i,D}: \\ \\\operatorname{supp}(d)<\sigma_d, \\ |d| \leq \maxlen, \\ \left\{\left\{k, k'\right\}: u_{k,k'}>u^{\mathrm{upper}}_{k,k'}\right\}\nsubseteq d}} \frac{1}{\operatorname{supp}(d) \times\lvert d \rvert^2}, \quad i=1,\dots,n,
\end{equation}

The rationale behind this formulation of the score is that the algorithm will scan over all itemsets $d$ that are contained within the discrete features of the $i$th observation. Then, if any of these itemsets appears less than $\sigma$ times within the data set, the score of the $i$th observation will be augmented in a way that is inversely proportional to the product of the support and the square of the length of $d$. Therefore, infrequent itemsets of smaller length are considered more likely to be outliers. Similar formulations of the discrete score were proposed by \cite{otey2006fast} and then by \cite{koufakou2010}. Moreover, we apply support-based pruning, so that supersets of infrequent itemsets are ignored in subsequent computations to reduce the computational cost. What this means in practice is that if an observation $\boldsymbol{X}_i$ includes a sequence of categorical variable levels $d_1$ which appears less than $\sigma_{d_1}$ times, then the contribution of any other sequence $d_2 \subseteq \boldsymbol{X}_{i,D}$ with $d_1 \subset d_2$ (and therefore $\lvert d_1 \lvert < \lvert d_2 \lvert)$ in $s_{D,(i,\cdot)}$ will be immediately set equal to zero.

Looking at Expression \eqref{eq:discretescorefinal}, we observe that it depends on the values of the parameters $\maxlen$ and $\sigma_d$. We recommend the use of a data-driven method for estimating reasonable values for the parameters. Starting with the threshold value $\sigma_d$, we would expect that in a completely uniform setting, each itemset has equal support. For instance, if the itemset that is being considered involves two discrete variables $j$ and $j'$, with $\ell_j$ and $\ell_{j'}$ levels respectively, then its support should be approximately equal to $n/\left(\ell_j \ell_{j'}\right)$, where $n$ is the number of observations. We could look at the occurrence of each itemset as a realisation of a Multinomial random variable with one trial (also known as the categorical or the multinoulli distribution), and set the number of events equal to the amount of itemsets that can be observed for the categorical variables involved. We can then construct simultaneous $100(1-\alpha)\%$ confidence intervals for the Multinomial probabilities of a $\mathrm{Multinomial}\left(1, \left\{ \left(n/\prod_{j \subseteq d}\ell_j\right), \dots, \left(n/\prod_{j \subseteq d}\ell_j\right) \right\} \right)$ random variable. Notice that $\prod_{j \subseteq d} \ell_j$ refers to the product of the levels of all categorical variables $j$ which are being considered in itemset $d$ and the vector of probabilities is of length $\prod_{j \subseteq d} \ell_j$. The construction of simultaneous confidence intervals for the event probabilities is done using the method described in \cite{sison1995simultaneous}. The confidence intervals will be identical, due to the expected proportions also being equal. Hence, once the confidence intervals have been constructed, we can take the lower bound of just one of them and multiply it by the number of observations $n$ to get a threshold value for the support of any itemset $d$ that involves the categorical features considered. We decide to set $\alpha = 0.01$ in all our simulations presented here, in order for the discrete score to be augmented only for observations with very infrequent levels or combinations of levels; a greater $\alpha$ value makes the algorithm more `conservative' by penalising less infrequent itemsets. This also helps alleviate issues caused by mild class imbalance, as it may just happen that a categorical level is observed less frequently than the rest, without necessarily being a sign of outlyingness; of course in case a certain level is extremely infrequent within a categorical variable, all observations including this will have an increase in their discrete score.

As we consider itemsets of greater length, the number of combinations of categorical levels that can be observed, which we denoted earlier by $\prod_{j \subseteq d}\ell_j$, will grow very large. This leads to a sparse representation of the data and consequently, we may encounter $\sigma_d<2$ for some itemset $d$. As a result, every sequence of levels of the categorical variables involved in $d$ will be considered as frequent and thus, there is no point in looking at itemsets of greater length. Therefore, we choose $\maxlen$ to be the largest number of categorical features for which any randomly chosen combination of $\maxlen$ discrete variables will yield $\sigma_d \geq 2$, where $d$ is any itemset of length $\maxlen$. This can be expressed as:
\begin{equation*}\label{eq:maxlendef}
    \maxlen = \max\left\{M: \sigma_d \geq 2 \ \mathrm{for \ all } \ d \subseteq\boldsymbol{X}_{i,D} \ \mathrm{with} \ \lvert d \rvert \leq M \right\}.
\end{equation*}

The final component of Expression \eqref{eq:discretescorefinal} consists of accounting for the association among discrete features. Strongly correlated discrete variables can be viewed to contain enough information for each other so that we can predict the level of one given the other. As a result, the combination of specific discrete levels may not be so frequent within the data set, without necessarily being outlying. A solution to this problem requires quantifying the nominal association between pairs of categorical features. We use the Uncertainty Coefficient (also known as Theil's $U$ coefficient) \citep{theil1970estimation} for this task; the $U$ score between two categorical variables $j$ and $j'$ is given by:

\begin{equation*}
    u_{j,j'} \defeq U\left(\boldsymbol{X}_{D_j} \mid \boldsymbol{X}_{D_{j'}}\right) = \frac{H\left(\boldsymbol{X}_{D_j}\right)-H(\boldsymbol{X}_{D_j}\lvert \boldsymbol{X}_{D_{j'}})}{H(\boldsymbol{X}_{D_j})},
\end{equation*}
where $H\left(Y\right)$ denotes the Shannon entropy of a discrete random variable $Y$ \citep{shannon1948mathematical}. Theil's $U$ score is not symmetric, meaning that in general $u_{j,j'} \neq u_{j',j}$, and it is affected by a different number of levels of the categorical features $j$ and $j'$. The upper threshold value for $u_{j,j'}$ is determined via a simulation strategy that considers pairs of distinct categorical variables $j$ and $j'$, samples $n$ realisations from a bivariate zero-mean Gaussian distribution with the correlation of the two features equal to 0.35 and then uses quantile discretisation with $\ell_j$ and $\ell_{j'}$ levels respectively. The values of $u_{j,j'}$ and $u_{j',j}$ are computed, their maximum is chosen and the process is repeated a total of fifty times. The mean of the maxima obtained is finally used as an upper threshold value $u_{j,j'}^\mathrm{upper}$; if $u_{j,j'}>u_{j,j'}^\mathrm{upper}$ or $u_{j',j}>u_{j',j}^\mathrm{upper}$, any itemsets $d$ of length at least equal to two that include the variables $j$ and $j'$ are set not to contribute in $s_{D,(i,\cdot)}$. This procedure is developed based on the fact that a higher correlation (in absolute value) prior to discretisation yields a greater $U$ score once the categorical features have been generated. The choice of a correlation of 0.35 between the two components of the bivariate Gaussian ensures that any association among discrete features is sufficiently weak and even a correlation of 0.4 can lead to an inflation of the discrete scores, based on an empirical study that we conducted. This additional restriction is referred to as the `correlation correction'.

The aforementioned scheme for calculating scores of outlyingness for discrete variables can be taken one step forward. More precisely, not only will we be computing the discrete score for each observation, but we will also be calculating the contribution of each categorical variable to this score. This can facilitate understanding of the source of outlyingness for marginal outliers detected by their discrete features. For any discrete variable $j$, its contribution to the discrete outlier score of an observation $i$ is given by:

\begin{equation}\label{eq:discretescorevarfinal}
    c_{D,(i,j)}=\sum_{\substack{d \subseteq \boldsymbol{X}_{i,D}: \\ j \subseteq d, \\\operatorname{supp}(d)<\sigma_d, \\ |d| \leq \maxlen, \\ \left\{\left\{k, k'\right\}: u_{k,k'}>u^{\mathrm{upper}}_{k,k'}\right\}\nsubseteq d}} \frac{1}{\operatorname{supp}(d) \times\lvert d \rvert^3}, \quad i=1,\dots,n, \quad j=1,\dots,p_D.
\end{equation}

Expression \eqref{eq:discretescorevarfinal} looks very similar to Expression \eqref{eq:discretescorefinal} but it has two main differences. The first one is that we are summing over sequences $d$ for which it holds that $j \subseteq d$; hence, any combination of discrete variables that yields an increase of the discrete score of observation $i$ will only be considered if the discrete variable $j$ is included in this combination. Secondly, the denominator now includes the cube of the length of the sequence $d$; this is because of the equal contribution of each discrete variable included in $d$, which is given by $1/\lvert d \rvert$. In fact, the contribution of each discrete variable to the discrete score of each observation can be stored in a $\left(n \times p_D\right)$-dimensional matrix with $(i,j)$th entry given by $c_{D,(i,j)}$. This matrix has the property that the sum of the $i$th row is equal to $s_{D,(i,\cdot)}$ and up to $\maxlen$ elements in each row can be non-zero.

We calculate the discrete scores for an artificial data set with a thousand observations, five discrete and five continuous features using Expression \eqref{eq:discretescorefinal}. We plot these against the continuous scores for each observation (calculated using the methodology described in the following subsection) on what we call a `score profile' plot. The inliers all have a zero discrete score and marginal outliers are well-separated from inliers and joint outliers, thus being much easier to detect. The experimental design is described in more detail in Appendix \ref{appendix:design}.
\vspace{-0.25cm}
\begin{figure}[h!]
  \centering
  \includegraphics[scale=0.6]{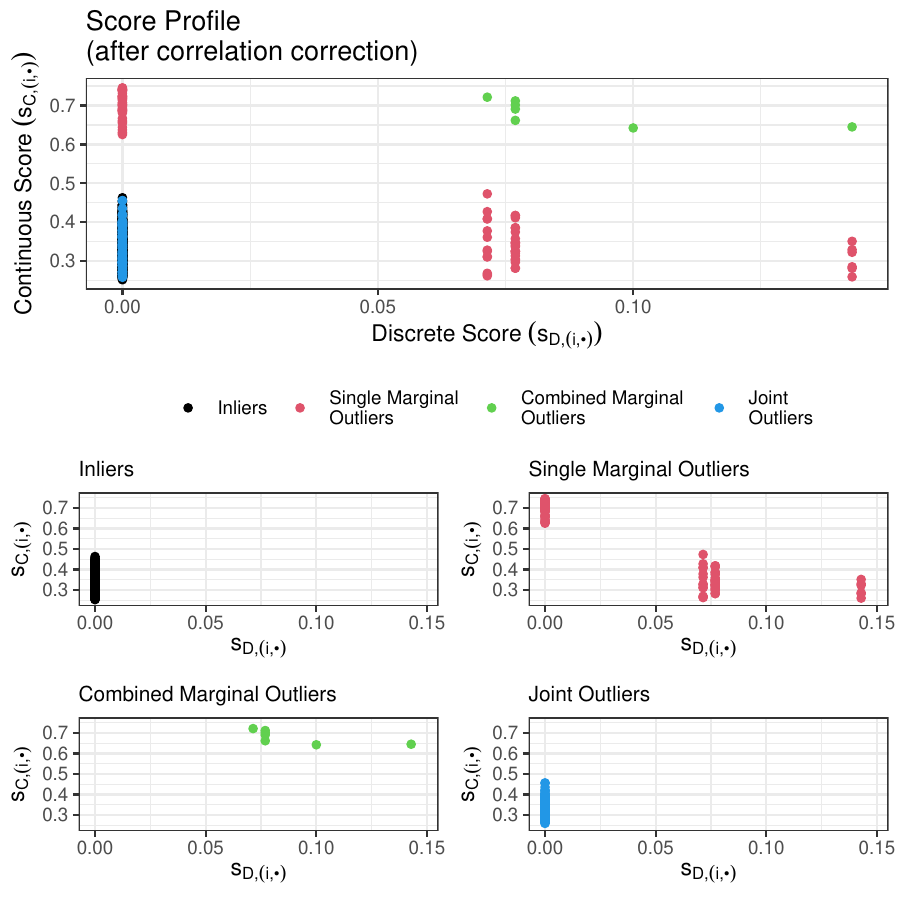}
  \vspace{-0.1cm}
  \caption{Score profile for artificially generated data set. Discrete scores have been computed using Expression \eqref{eq:discretescorefinal}.}
  \label{fig:score_profile_after}
\end{figure}
\subsection{Continuous Score}

Our definition of outliers in the continuous space states that a point is marginally outlying in the continuous domain if it is far away from the rest of the observations. This motivates the use of the Isolation Forest algorithm of \cite{liu2008isolation} for computing the continuous score of outlyingness for an observation, which we briefly describe here. The main idea of the Isolation Forest algorithm is that anomalies are few and easy to isolate from the rest of the data, so in multiple random partitionings of the observations of a data set, the anomalies will require a smaller number of splits to be isolated.

An Isolation Forest is built using an ensemble of Isolation Trees, where each Isolation Tree represents a random partitioning of the data. More precisely, the algorithm randomly selects a sample of the observations and a continuous feature $j$ from the set of continuous variable $\boldsymbol{X}_C$, on which a split is made. A random value $\psi$ in the range $\left[ \min \boldsymbol{X}_{C_j}, \max \boldsymbol{X}_{C_j} \right]$ is selected and the observations are partitioned in two sets, the first containing these data points with $j$th continuous variable value greater than $\psi$ and the second containing the rest. This process continues with a different continuous variable chosen for each split, up to a certain number of times. The partitioning of the data using these splits can be represented by a binary tree structure, where each parent node is a condition that defines a split on some continuous variable and the two daughter nodes include the points that satisfy the condition and these that do not satisfy it, respectively. An external node with no children means that one point has been isolated (and can therefore not be split any further), while the amount of times the splitting procedure is repeated is the height of the tree. The tree height is by default chosen to be equal to $\lceil \log_2 n_s \rceil$, where $n_s$ is the sub-sampling size and $\lceil \cdot \rceil$ is the ceiling function. This is motivated by the fact that the average tree height of a binary tree constructed from $n_s$ observations is equal to $\lceil \log_2 n_s \rceil$, hence any data points which have not been isolated after $\lceil \log_2 n_s \rceil$ splits are more likely not to be anomalous.

In order to construct an Isolation Forest, we need to grow multiple Isolation Trees on different sub-samples of the data. The rationale behind the use of sub-sampling is motivated by the problems of swamping and masking, corresponding to the wrong identification of inliers as outliers and the difficulty in detecting clustered anomalies, respectively. The number of sub-samples is empirically suggested by \cite{liu2008isolation} to be equal to $2^8$; a greater sub-sampling size increases the computational complexity of the algorithm. Once all Isolation Trees have been grown, the continuous score for an observation $i$ is calculated using:

\begin{equation}\label{eq:contscore}
    s_{C,(i,\cdot)} = 2^{-\frac{\mathbb{E}\left(h\left(\boldsymbol{X}_{i,C}\right) \right)}{c\left(n_s\right)}},
\end{equation}
where $\mathbb{E}\left(h\left(\boldsymbol{X}_{i,C}\right) \right)$ is the average depth of an observation $i$ in all the Isolation Trees grown. An anomaly is more likely to need a smaller number of splits until it is isolated, thus its average depth will be lower and this will clearly lead to a greater continuous score. The normalising constant $c\left(n_s\right)$ is the average depth in an unsuccessful search in a Binary Search Tree and it is equal to $c\left(n_s\right)=2H_{n_s-1}-2\left(n_s-1\right)/n_s$, where $H_{n_s}$ is the harmonic number and it is approximated by $\log{n_s}+\gamma$, with $\gamma\approx 0.57722$ being the Euler-Mascheroni constant \citep{bruno1999data}.

We make use of an extension of the Isolation Forest algorithm, called the Extended Isolation Forest \citep{hariri2019extended}. The main difference of the Extended Isolation Forest to the original Isolation Forest algorithm is that instead of defining splits across individual continuous variables, linear combinations of continuous variables are considered. The splits are no longer defined by hyperplanes which are parallel to the axis corresponding to the continuous variable on which the split is being made; inclined hyperplanes are used to partition the data into two regions, thus defining a split that takes more than just one continuous variable into account. This results in more accurate continuous score values, with some illustrative examples given in \cite{hariri2019extended}. For the rest of the paper, it can be assumed that we make use of the Extended Isolation Forest algorithm for continuous outlier detection, with a sub-sample size $n_s$ equal to $2^8$, 500 trees, a maximum height of 100 and all $p_C$ continuous variables are used for defining the splits. The choice of a maximum tree height of 100 may come in contrast to the recommendation of \cite{liu2008isolation} who suggest a maximum height of $\lceil\log_2n_s\rceil$ but we want to ensure that the only observations that take large $s_{C,(i,\cdot)}$ values are those that deviate significantly from the rest of the data points in the continuous domain.

\section{Detection of marginal outliers}\label{sec:detectmargouts}

Detecting the marginal outliers in a mixed data set is a key step of the process of identifying anomalies. Successful detection of the marginal outliers is essential for subsequent parts of the procedure and especially for the detection of joint outliers; having a big amount of marginal anomalies in the data may mask some interesting relationships among discrete and continuous variables. In this section, we present a method that can be used to flag atypical observations in the discrete space, as well as the strategy that is implemented for outlier detection in the continuous space. The methods presented are based on the score values $s_{D,(i,\cdot)}$ and $s_{C,(i,\cdot)}$ and are supported by properties of these scores which have been made upon visual inspection of score profiles for data sets with varying proportions of outliers, number of discrete levels for categorical variables and number of observations. We finally conduct a simulation study, from which it can be seen that our methods work well and can correctly flag the big majority of marginal outliers.

\subsection{Discrete space}

One of the main remarks that can be made when looking at the discrete scores is that they only take a small number of values, with many of them being exactly zero. This comes as no surprise, as in a data set of $n$ observations, we could potentially have $n$ distinct discrete scores in total. However, the fact that only a small proportion of points is assumed to be marginally outlying in the discrete space means that only a few of the observations are going to have a non-zero discrete score. Even the non-zero discrete scores corresponding to inliers (or joint outliers) are most likely small in value, since these should be equal to $1/(\mathrm{supp}(d) \times \lvert d \lvert^2)$, where $\lvert d \lvert$ is either equal to $\maxlen$ or to a value slightly smaller than that, thus yielding a small discrete score value, based on our definition of marginal outliers in the discrete space.

According to this observation, we implement $K$-Means on the set of discrete scores. In order to reduce the computational time and power required, we can consider the unique discrete score values obtained and perform one-dimensional $K$-Means on these. The choice of $K$ ranges from 1 up to the number of unique discrete scores, which is small enough to ensure a rapid implementation. Moreover, we know that there will be a cluster including points with zero score when running $K$-means, which we denote by $\mathcal{C}_0^K$. Running $K$-Means with $K=1$ implies that all data points will be assigned into $\mathcal{C}_0^1$. As $K$ increases, we expect the cluster that contains the most distant scores to be split, so as to ensure that the within sum of squares of each of the resulting two clusters (after the split) is reduced. Therefore, any low discrete scores that are very close to zero (thus corresponding to inliers) will remain in $\mathcal{C}_0^K$, until $K$ gets large enough. Looking at how the size of $\mathcal{C}_0^K$ varies with $K$ can then give us an indication of the number of clusters that need to be considered to flag marginal outliers. More precisely, we look at consecutive values of $K$ for which the size of $\mathcal{C}_0^K$ remains constant and greater than $n-\lceil (\rho + \epsilon )n \rceil$ upon removing any observations with an infrequent sequence of unit length from $\mathcal{C}_0^K$. Here, $\lceil \cdot \rceil$ refers to the ceiling function, with $\lceil x \rceil$ returning the smallest integer $x'$ such that $x' \geq x$, while $\rho$ is the maximum proportion of outliers we believe is included in the data and $\epsilon$ is an additional proportion of anomalies that we are willing to tolerate (with $\rho + \epsilon \leq 0.50$). We choose $\rho = 0.20$ and $\epsilon=0.02$ throughout the simulations presented in this paper, hence $0.22n$ as the maximum number of discrete outliers.

We further ensure that any observations with an infrequent level are flagged as marginal outliers, as these are in practice more likely to be outlying than e.g. observations with infrequent itemsets of a greater length. Recalling the definition of outliers in the discrete space as observations that include some highly unlikely discrete level or sequences of discrete levels, we expect to see a substantial difference between the minimum score of a marginal outlier in the discrete space and the maximum score of an inlier with a non-zero score, caused by the occurrence of itemsets $d$ slightly less times than their respective thresholds $\sigma_d$. In order to spot this difference, we look at consecutive differences of the absolute values of scaled scores (`scaled' here refers to the scores being shifted by the mean of all discrete scores and divided by their standard deviation). We define a difference in the scaled scores to be significant if it exceeds a unit, which we empirically find to be a good threshold. Notice that one can also obtain analytic expressions for the upper and lower bounds for the mean and the standard error of the discrete scores (see Propositions \ref{prop:avgdiscscore} and \ref{prop:sddiscscore} in Appendix \ref{appendix:proofs}). Once we have found these intervals of $K$ values for which $\lvert \mathcal{C}_0^K \rvert$ is constant, with $\lvert \mathcal{C}_0^K\rvert > n- \lceil (\rho+\epsilon)n \rceil$ upon removing any observations with scores less than the value for which a significant difference is detected, as well as observations containing an infrequent discrete level, we select the widest range of $K$ values for which all these conditions are satisfied and choose a value of $K$ from this range. We summarise the process in Algorithm \ref{code:marg_cat_outs} (Appendix \ref{appendix:algorithms}).

\subsection{Continuous space}

Once the marginal outliers in the discrete space have been flagged, we can discard these points from the process of detecting anomalous observations in the continuous space since their continuous scores are no longer relevant. Things become a bit more complicated now, since the continuous scores are most likely all going to be unique. However, the continuous scores obtained by the Extended Isolation Forest algorithm are all going to be relatively close for inliers; this follows easily by considering the way the scores are calculated by the Extended Isolation Forest algorithm (see Expression \eqref{eq:contscore}). Global outliers in the continuous domain will certainly need a much smaller number of splits to be isolated, thus leading to a much lower expected path length and as a result, to a greater score value. On the contrary, we expect that no inlier will need as few splits to be isolated, therefore the difference in score should be quite large between an inlier and an outlier. This is also expected to be the case for the difference between the score of the outlier that needs the most splits to be isolated and that of the inlier that needs the less splits to end in a terminal node. These two observations correspond to the outlier and the inlier with the lowest and the highest continuous scores respectively. Given that there is going to be a rather significant gap between these two scores, while the scores for the inliers are all going to be very dense, we seek to find large distances between the sorted continuous scores, in an attempt to identify the score value beyond which outliers are present.

In order to detect the large gaps between the sorted continuous scores, we look at the values of the differences between consecutive scores and find the values which are extreme. This is done using Chebyshev's Inequality \citep{chebyshev}, which states that for a random variable $Z$ with mean $\mathbb{E}(Z)$ and non-zero standard deviation $\sqrt{\mathrm{Var}(Z)}$, for any $\lambda \in \mathbb{R}^+$, it holds that $\mathbb{P}\left(\lvert Z - \mathbb{E}(Z) \rvert \geq \lambda \sqrt{\mathrm{Var}(Z)}\right) \leq 1/\lambda^2$. In our case, $Z$ is the random variable of the differences between consecutive sorted continuous score values and we define $m_C$ and $s_C$ to be the average and standard deviation of these differences, respectively. Since Chebyshev's inequality does not make any distribution-related assumptions on the random variable of interest, it can serve as a tool for detecting extreme values among the differences. Our approach is then straightforward; we consider integer $\lambda$ values in the range $[2,20]$ and we look at the amount of sorted differences $z_i$ satisfying $\lvert z_i - m_C \rvert \geq \lambda \times s_C$. This yields the set $\Delta_\lambda \vcentcolon= \left\{z_i : \lvert z_i - m_C \rvert \geq \lambda \times s_C \right\}$. As long as $\lvert \Delta_\lambda \rvert$ is constant and non-zero for many consecutive values of $\lambda$, that is an indication of some differences between consecutive score values being so large that only a big $\lambda$ value can further decrease $\lvert \Delta_\lambda \rvert$. Notice that observing differences that are below $20\times s_C$ can only happen with probability less than $0.25\%$, which is probably infinitesimally small enough to ensure that $\lvert \Delta_\lambda \rvert$ is constant for at least a few consecutive $\lambda$ values should any large gaps exist.

We plot the score profile for an artificial data set with 5000 observations, 8\% of which are marginal outliers in Figure \ref{fig:score_prof_gaps}, as well as a plot of $\lvert\Delta_\lambda\rvert$ against $\lambda$ for that data set. Figure \ref{fig:gaps_plot} reveals that $\lvert \Delta_\lambda \rvert$ keeps decreasing until it gets constant for $\lambda = 15$ and then drops by a unit for $\lambda\geq 19$. For $\lambda \in [15,18]$ we observe $\lvert\Delta_\lambda\rvert = 3$, meaning that two large gaps are detected. These are marked by the dashed lines in the score profile. The final step is to choose a lower threshold for the continuous score, beyond which our observations are flagged as outliers. We prefer to select the score value corresponding to the bottom line as long as this is above $0.4$ (it is very unlikely that an outlier will have a continuous score below 0.4 empirically). We also ensure that the total number of marginal outliers in the data set is again below $\lceil (\rho+\epsilon)n \rceil$. This is all summarised in Algorithm \ref{code:marg_cont_outs}.

\begin{figure}[h!]
  \centering
  \begin{minipage}[b]{0.475\textwidth}
    \includegraphics[width=\textwidth]{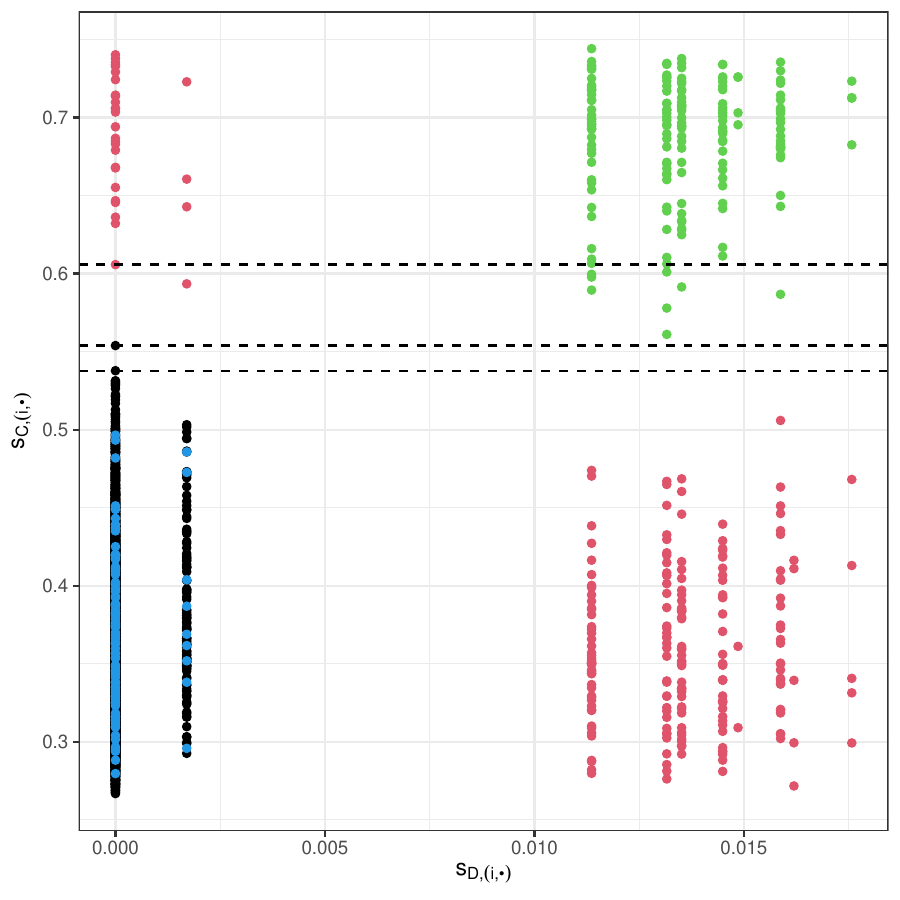}
    \caption{Continuous and discrete scores for data set; the dashed lines indicate the two large gaps between consecutive sorted continuous scores.}
    \label{fig:score_prof_gaps}
  \end{minipage}
  \hfill
  \begin{minipage}[b]{0.475\textwidth}
    \includegraphics[width=\textwidth]{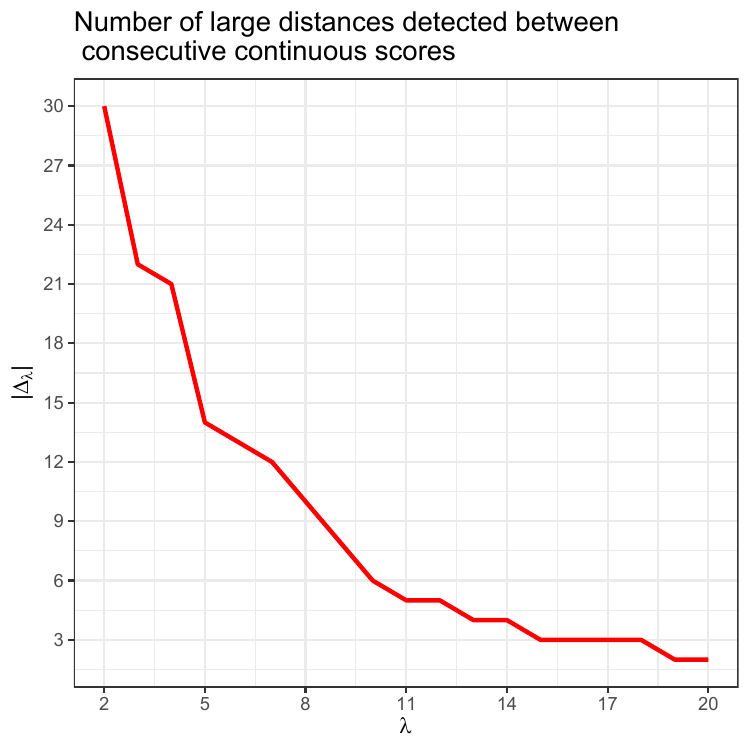}
    \caption{$\lvert \Delta_\lambda \rvert$ against $\lambda$; we see that $\lvert \Delta_\lambda \rvert$ remains constant for $\lambda \in [15,18]$, hinting two large differences between consecutive sorted continuous scores.}
    \label{fig:gaps_plot}
  \end{minipage}
\end{figure}

We have used Algorithms \ref{code:marg_cat_outs} \& \ref{code:marg_cont_outs} from Appendix \ref{appendix:algorithms} to detect the marginal outliers in artificial data sets with varying number of observations, proportion of marginal outliers and number of discrete levels. More precisely, we have experimented with data sets including one, three, five, seven and ten thousand data points. For each of these numbers of observations, we have considered proportions of outliers of 5\%, 10\%, 15\% and 20\% and for each proportion of outliers, the outliers in each data set consisted of 20\%, 50\% and 80\% marginal outliers (the rest being joint outliers). The categorical variables in each data set consisted of two up to seven discrete non-outlying levels. Notice that in all cases, the number of discrete levels for each categorical feature was set to be equal; letting this differ for each discrete variable would be interesting to check but it would lead to a prohibitively large number of simulations. A hundred artificial data sets were generated for each of these scenarios, except for the case of data sets with a thousand observations and seven discrete levels (the number of discrete levels is too large for such small data sets, thus the respective threshold values for itemsets of unit length lead to every single discrete level being considered as infrequent), leading to a total of 34800 data sets on which the two algorithms were tested. The results are summarised in terms of the recall and the F1 score; the recall is the proportion of marginal outliers which are correctly identified, while the F1 score also takes into account the number of non-marginal outliers that have been erroneously flagged as such.
\vspace{-0.4cm}
\begin{figure}[h!]
  \centering
  \includegraphics[scale=0.68]{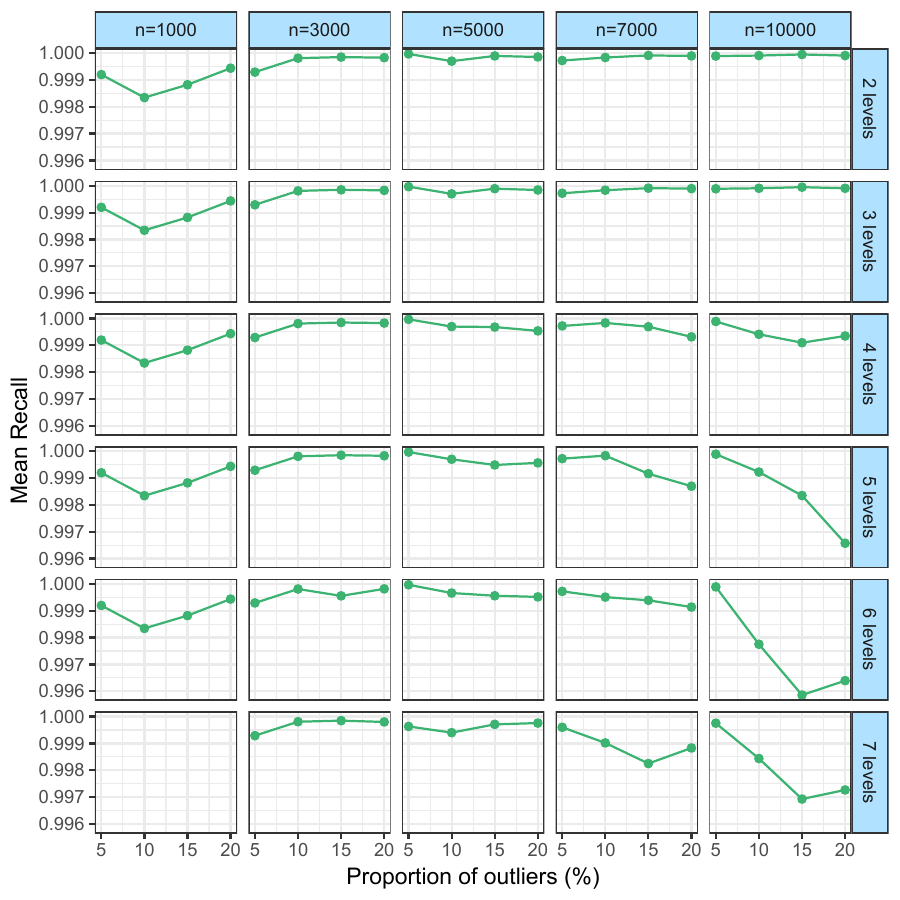}
  \vspace{-0.1cm}
  \caption{Mean recall for simulation study on detection of marginal outliers in artificial data sets for varying number of observations and discrete levels, as well as for increasing proportion of outliers.}
  \label{fig:res_marg_recall}
\end{figure}

As can be seen, our results in Figure \ref{fig:res_marg_recall} reveal that the recall is constantly above 99.5\%, hence the vast majority of marginal outliers are successfully flagged. The performance drops as the number of discrete levels increases, especially for larger data sets but the recall values remain high. A similar pattern is observed for the F1 score; Figure \ref{fig:res_marg_f1} in Appendix \ref{appendix:extrafigstables} shows that while a large number of marginal outliers are detected for large sample sizes, more inliers are also flagged as marginally outlying, as the number of levels increases. This comes as no surprise, as a larger number of observations yields a higher $\maxlen$ value, thus itemsets of greater length are considered, while a large number of discrete levels can lead to some infrequent sequences of length greater than a unit. However, the results remain at a reasonably good level, with not so many inliers being flagged as outlying and with almost every single marginal outlier being successfully detected.

\section{Detection of joint outliers}\label{sec:detectjointouts}

Working with mixed-type data imposes a significant challenge; dealing with the mixed-attribute space. While a plethora of algorithms exist for detecting anomalies in the continuous and the discrete space individually, we need to treat the mixed attribute domain carefully enough, so that we can flag outliers that may exist there. This is where we consider our definition of joint outliers. This definition is a special case of the concept of contextual outliers, which are observations that deviate from the rest of the data with respect to some context \citep{datamining2012}. In our case, the context is a set of continuous features, while the behavioural attribute (defining the characteristics of each observation with respect to the context) is one discrete variable at a time. We give an intuitive explanation to ease understanding. Assume a discrete feature is associated with a set of continuous variables so that projecting the data in the space spanned by this set and colouring each point according to the level of the discrete feature involved reveals a separation between the different levels. Any observation that is spotted to be in the `wrong region' based on its discrete variable level is then flagged as a joint outlier. In this section, we describe a method for finding any such associations and flagging observations that do not conform to it. We also introduce a method that aids the choice of suitable threshold values related to the process.

\subsection{Finding associations between discrete and continuous features}

A question that naturally arises from the definition of joint outliers is whether and how we can find any association among a categorical and a set of continuous features. We approach this by viewing the joint outliers as contextual anomalies, reducing our problem to that of context identification in contextual outlier detection.

Identifying the context with respect to which there exists a data structure is a problem that becomes increasingly difficult in the presence of variables which are not associated in any way with a target discrete feature of interest. \cite{liang2016robust} use ideas of `contextual neighbours' and `local expected behaviour' for identifying contextual outliers but these require knowing the context a priori, while \cite{zhang2016relevant} consider the sparsity of the variables for high-dimensional data sets as a context identification step. However, both approaches are restricted on contextual and behavioural features being continuous only. To the best of our knowledge, the only available algorithm for automatic context identification which can handle different types of data is the ConOut algorithm of \citep{meghanath2018conout}. Specifically for the case of a discrete behavioural variable and a set of continuous contextual features, they propose using the Kruskal-Wallis $H$ test \citep{kruskal1952use}. We consider this as a starting point and extend the idea by including an additional stage of testing.

The Kruskal-Wallis $H$ test checks for differences in average group ranks, where `group' refers to each level of the target discrete feature in this case. The use of ranks is a desirable property here due to the presence of outliers with respect to each level, thus making actual values prohibitive for the analysis, while the non-parametric nature of the test does not impose any distributional restrictions either. Fixing a target discrete variable $\boldsymbol{X}_{D_j}$ and having removed any marginal outliers, we perform this test for each continuous feature individually and compare the resulting $p$-value with a pre-specified significance level $\alpha_1$. Any continuous feature for which the test is rejected is considered to be a candidate belonging to the final set of contexts, which we denote by $\mathcal{J}$.

%Checking if the average ranks differ serves as a first indication of a structure between the different levels of $\boldsymbol{X}_{D_j}$ but it is not enough on its own.
While the $H$ test returns a list of candidate context features $\mathcal{J}$, we need to ensure that the levels are still `well-separated' with respect to the final context. Having a good level of `separation' in this case can be translated as the `core' of each level being surrounded mostly by observations of that same level, where the `core' refers to the most centrally-located point within the level with respect to some similarity measure. The choice of the similarity measure is defined in a manner that takes the geometry of the observations into account.

Should no structure exist, we expect that for any given observation, the proportion of its nearest neighbours of any level $l$ with respect to any distance metric will be roughly equal to that of the occurrence of level $l$ in the data set. Therefore, we can formulate the problem as a hypothesis test of the $\lceil \delta \times n_{j,l} \rceil$-nearest neighbours of the core of level $l$ following a Multinomial distribution with $\lceil \delta \times n_{j,l} \rceil$ trials and probability vector $\boldsymbol{\pi}_j = \left(\pi_{j,1}, \ldots, \pi_{j, \ell_j}\right)$. The value of $\delta$ is chosen to be at most equal to 0.50, as we seek to detect the most centrally located points of the $n_{j,l}$ observations for which $\boldsymbol{X}_{D_j} = l$. Moreover, $\pi_{j,l}$ is the proportion of these observations in the data set upon the removal of marginal outliers. We perform a Pearson's chi-square goodness of fit test for each of the $\ell_j$ levels independently and obtain $\ell_j$ $p$-values. The $p$-values are then compared to a pre-specified significance level $\alpha_2$, where the Holm-Bonferroni method \citep{holm1979simple} is applied to account for multiple testing considerations. If all $\ell_j$ tests reject the null hypothesis, then we can assert that a structure has been identified.

As mentioned earlier, the core of each level is detected with respect to a similarity measure that considers the geometry of the observations. More precisely, our method uses a weighted Minkowski distance, scaled by the eigenvalues obtained by performing the Robust Principal Component Analysis (ROBPCA) method of \cite{hubert2005robpca} to account for the variability along each dimension. The final step of the procedure is to check whether and which continuous variables need to be removed from the set of contextual features $\mathcal{J}$, where the product of the $\ell_j$ $p$-values for each subset of $\mathcal{J}$ determines what the final context should be (see Proposition \ref{prop:minpvalue} in Appendix \ref{appendix:proofs} for a justification in the context of the given problem). This procedure of combining $p$-values was introduced by \cite{zaykin2002truncated} and we refer the reader to the original paper for a more theoretical explanation. The process used for context identification is summarised in Algorithm \ref{code:context_id_1d} for one dimension and its extension in higher dimensions is included in Algorithm \ref{code:context_id_hd} (Appendix \ref{appendix:algorithms}). Notice that Algorithm \ref{code:context_id_hd} describes a backward elimination approach for determining $\mathcal{J}$; if this fails for the initial set of candidate features, a forward selection process (starting from pairs of features) is used instead.

We finally introduce three different types of relationships between the first discrete and a set of two continuous variables for illustrative purposes; these are referred to as the `Linear', the `Product' and the `Quotient' designs and are used for simulation study purposes. For the linear design, we define the levels of the first discrete variable by quantile discretisation of the difference of the values of the continuous features involved. Similarly for the product and the quotient designs, we discretise the values of the product and the quotient of these continuous variables, respectively. These designs, depicted in Figure \ref{fig:designs_2d_scatter}, are chosen to resemble some real-world examples (e.g. the Body Mass Index (BMI) which is computed as a ratio that leads to a classification into four categories, or the Air Quality Index (AQI), which also involves linear operations).

\vspace{-0.25cm}
\begin{figure}[h!]
  \centering
  \includegraphics[scale=0.75]{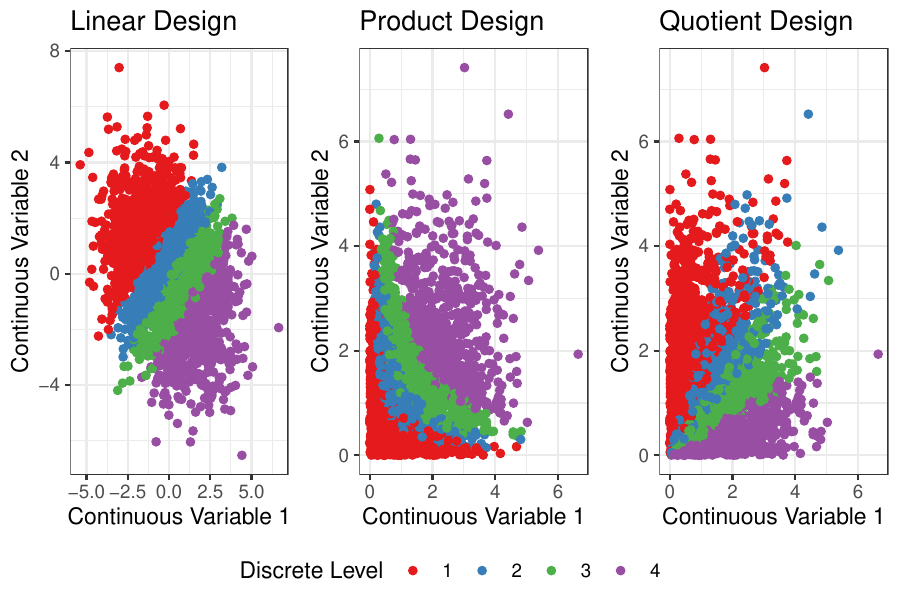}
  \vspace{-0.1cm}
  \caption{Illustration of the Linear, the Product and the Quotient designs on artificial data sets with 3000 observations, out of which 8\% are marginal and 2\% are joint outliers. The relationships are defined between the first two continuous features and the first discrete variable, consisting of four non-outlying discrete levels. The marginally outlying points are not plotted.}
  \label{fig:designs_2d_scatter}
\end{figure}

\subsection{Detecting joint outliers via density estimation}

Having identified whether any discrete variables in the data are associated with any set of continuous features, we can proceed by detecting the joint outliers. We will assume (without loss of generality) that there exists a relationship between the $j$th discrete variable $\boldsymbol{X}_{D_j}$ and the set of continuous features $\boldsymbol{X}_{C_\mathcal{J}}$, where $\mathcal{J}$ is a set of indices for the continuous features associated with $\boldsymbol{X}_{D_j}$. The joint outliers are the observations that do not conform to the existing data pattern, so that if we were to project the observations in the space spanned by $\boldsymbol{X}_{D_j}$ and colour them according to $\boldsymbol{X}_{D_j}$, then we would observe them in the wrong regions, meaning that they would be located in regions of high density of observations possessing a different level. Thus, flagging these outlying observations can be viewed as a density estimation problem.

The general idea of density estimation is that if we have multiple observations from some unknown distribution with density $f$, we seek to approximate $f$ using the data that is available to us. In our case, given the realisations $\boldsymbol{X}_{C_\mathcal{J}}$, our goal is to estimate $\ell_j$ densities, corresponding to the $\ell_j$ discrete levels of $\boldsymbol{X}_{D_j}$. We can use Kernel Density Estimation (KDE) to accomplish this. Given $n$ observations, the kernel density estimator for a density of interest $f$ is given by:
\begin{equation}\label{eq:kde}
    \hat{f}(x) = \frac{1}{nh}\mathlarger{\sum}\limits_{i=1}^n K\left(\frac{X_i-x}{h} \right),
\end{equation}
where $h$ is the bandwidth and $X_1, \dots, X_n$ are the realisations of a univariate random variable with density $f$ (the extension to a multivariate setting follows straightforwardly). The function $K$ is a kernel function, which is a non-negative function that satisfies some moment conditions \citep[for a more extensive discussion, see][]{parzen1962estimation}. A more detailed derivation of Expression \eqref{eq:kde} can also be found in \cite{loader2006local}.

Implementing KDE requires the choice of a value for the bandwidth, as well as of the kernel function, with the former being of vital importance if bias or variance levels are of interest. However, \cite{loader2006local} argues that the effect of the bandwidth value, as well as of the degree of the polynomial used in local regression (this is used to approximate the density function in a small neighbourhood of a point $x$) is small when KDE is to be used for classification purposes. Nevertheless, we choose an adaptive nearest neighbor bandwidth to avoid fitting problems caused by data sparsity. More precisely, pairwise distances between the fitting point $x$ and each observation $x_i$ are computed and the bandwidth is then chosen as the $k$th smallest distance, with $k=\lfloor n\alpha \rfloor$, where we set the smoothing parameter $\alpha$ equal to 0.3. The kernel function does not influence the classification results on a large scale either, so we proceed with the Gaussian kernel. A final but crucial remark is that KDE only makes use of the continuous variables $\boldsymbol{X}_{C_\mathcal{J}}$; the use of redundant features will lead to poor classification results, which highlights the great importance of correctly detecting the set of continuous features associated with $\boldsymbol{X}_{D_j}$.

Once the kernel density estimators for all $\ell_j$ levels of $\boldsymbol{X}_{D_j}$ have been constructed, we can do classification of each of the observed data points. We are not interested in the observations that were previously marked as marginal outliers, so we only classify the remaining observations in an attempt to detect the joint outliers. The classification rule here consists of looking at the level yielding a maximal density estimate at each point; thus for an observation $i$, the predicted level $\ell^*$ is given by:
$$
\ell^* = \argmax_{l=1,\dots, \ell_j} \hat{f}_l \left(\boldsymbol{X}_{D_j} \mid \boldsymbol{X}_{i,C_\mathcal{J}} \right),
$$where $\hat{f}_l(\cdot)$ denotes the KDE of the density of the $l$th level and $\boldsymbol{X}_{i,C_\mathcal{J}}$ represents the vector of values of the continuous variables indexed by $\mathcal{J}$ for the $i$th observation. By our definition of the joint outliers, in an ideal scenario we would expect these to be misclassified and that would complete the process of detecting outliers in the mixed-attribute domain. However, in many practical applications the different levels will be overlapping, such as what is shown in Figure \ref{fig:designs_2d_scatter} and these require some extra care.

In order to alleviate the aforementioned shortcoming, we look at the ratio of the maximum KDE value to that of the KDE for the true level of each observation that is misclassified. More precisely, assuming that the $i$th observation has true level $l^{\text{true}}$ but it is misclassified, we define the following ratio:
\begin{equation}\label{eq:kderatio}
\Lambda_i = \frac{\max\limits_{l=1, \dots, \ell_j}\hat{f}_l\left(\boldsymbol{X}_{D_j} \mid \boldsymbol{X}_{i, C_\mathcal{J}}\right)}{\hat{f}_{l^\mathrm{true}}\left(\boldsymbol{X}_{D_j} \mid \boldsymbol{X}_{i, C_\mathcal{J}}\right)}
\end{equation}
The ratio in Expression \eqref{eq:kderatio}
will always be greater than a unit under the assumption that the $i$th observation has been misclassified; had this not been the case, we would get $\Lambda_i=1$. Now inliers close to the boundary regions of the density that corresponds to their discrete level are expected to be misclassified but the corresponding value of Expression \eqref{eq:kderatio} is not expected to be too large. As a result, setting a threshold $\Lambda^*$ for $\Lambda_i$ could serve as a criterion regarding which observations are actually outlying in the mixed-attribute domain and which are just inliers; any observations for which $\Lambda_i > \Lambda^*$ will be treated as joint outliers and the rest of the misclassified points will be considered inliers.

One way of finding a suitable threshold value $\Lambda^*$ is by looking at the number of misclassified observations for which $\Lambda_i > \Lambda^*$ as $\Lambda^*$ varies. We expect this number to drop sharply as $\Lambda^*$ increases, with a significant drop for a small $\Lambda^*$ value and then reaching a plateau. The large drop is an indication of the inliers being treated as such if we allow for the threshold value of Expression \eqref{eq:kderatio} to be large enough, while the joint outliers, being observations that lie in completely wrong regions in the space spanned by $\boldsymbol{X}_{C_\mathcal{J}}$, should yield much larger ratio values. This is displayed graphically in Figure \ref{fig:Kplots}, where we have plotted the amount of misclassified observations satisfying the criterion $\Lambda_i > \Lambda^*$ for a range of $\Lambda^*$ values from one to twenty, in steps of half a unit (smaller steps could have been considered but that would increase the computational cost) for three different cases. These three plots were generated from artificial data sets consisting of 3000 observations, with the discrete variables including four discrete levels. An association between the first discrete and two continuous features was imposed using the product design, in order to define the joint outliers. The top left plot (Case 1) corresponds to a data set generated as described above with 5\% of its observations being outliers, 20\% of which are joint outliers (the remaining 80\% are marginal outliers). For the second and third data sets (Cases 2 and 3), we have contaminated 20\% of the observations, 50\% and 80\% of which are joint outliers, respectively.
\vspace{-0.5cm}
\begin{figure}[h!]
  \centering
  \includegraphics[scale=0.3]{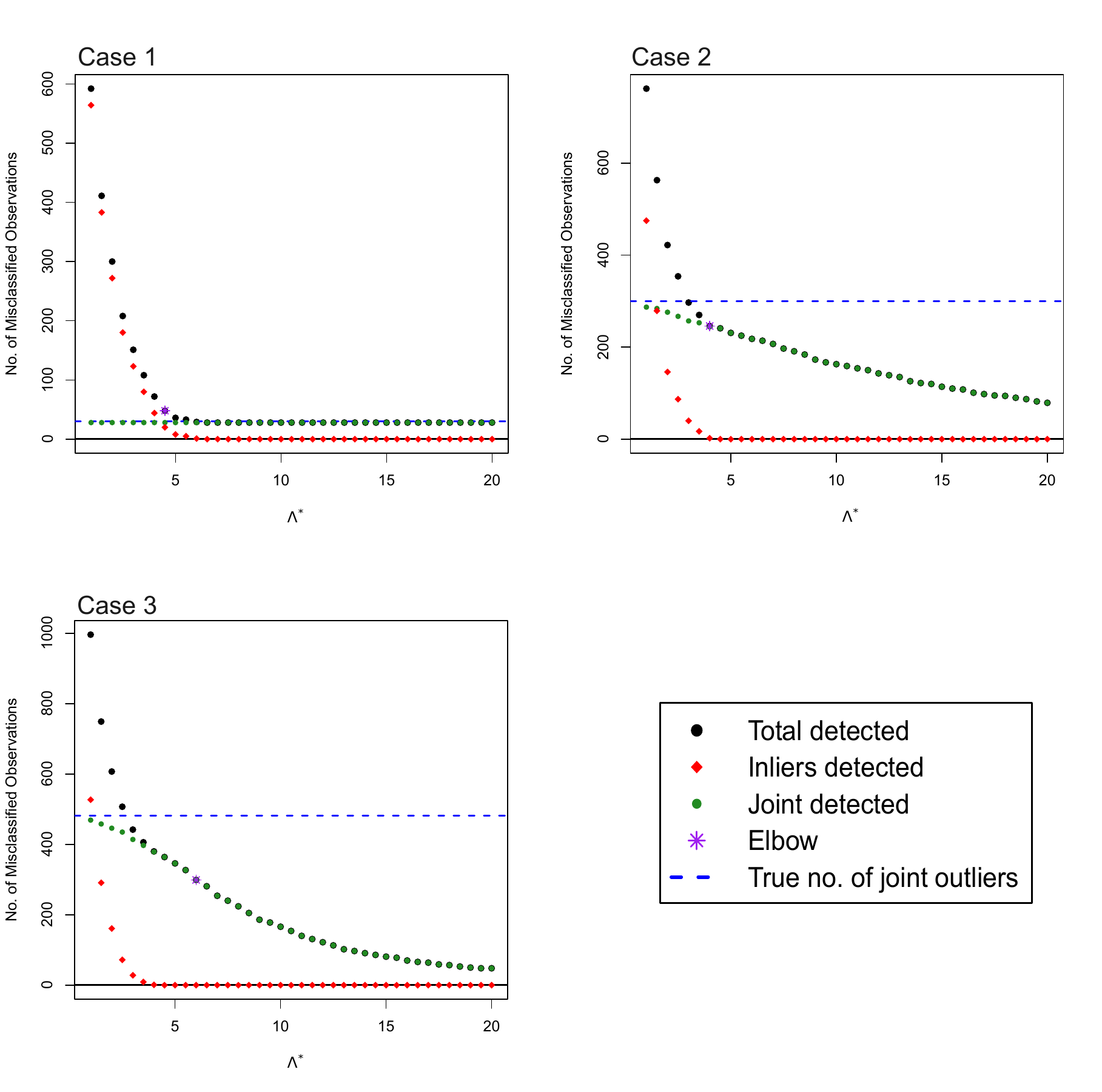}
  \vspace{-0.1cm}
  \caption{Number of misclassified observations satisfying $\Lambda_i > \Lambda^*$ against $\Lambda^*$ for the three cases described. The black dots correspond to the total amount of misclassifications, with the red and green points being the inliers and joint outliers included in this total, respectively. The blue dashed line shows the true number of joint outliers in the data set and the purple asterisk is the elbow point of the curve defined by the total number of misclassifications.}
  \label{fig:Kplots}
\end{figure}

Figure \ref{fig:Kplots} reveals some interesting patterns, as well as some issues that we may encounter when developing a strategy for choosing $\Lambda^*$. For instance, we can see the number of inliers (red diamonds) drops significantly as $\Lambda^*$ increases, which was expected. The number of joint outliers (green points) does not drop that sharply and is in fact almost constant (see instance in Case 1), although that does not hold for Cases 2 and 3. Moreover, although a lower value of $\Lambda^*$ closer to a unit can ensure a large detection rate of joint outliers, the number of erroneously flagged inliers is very large. We have also plotted the elbow point (purple asterisk) for each of the three cases. Detecting the elbow based on discrete data is achieved using the Kneedle algorithm of \cite{satopaa2011finding}, who define the knee of a graph as the point of maximum curvature when a curve is rotated in such a way that the minimum and maximum values on the vertical axis are co-aligned. Despite the Kneedle algorithm having been developed to detect the knee of a concave increasing curve, it can also be used to find the elbow when the curve is convex and decreasing (such as in our case). In a completely unsupervised setting, we would only have access to the total number of misclassified points (the black dots), so we could look at the elbow of the graph of misclassifications, making the Kneedle algorithm could be of great use. However, this algorithm fails to return the optimal $\Lambda^*$ value in some cases; in Case 1, it underestimates the value of $\Lambda^*$, leading to a bunch of inliers being treated as joint outliers, whilst in Case 3, there is no clear elbow due to the maximum curvature not being sufficiently large. Therefore, the method overestimates $\Lambda^*$ and yields a very conservative threshold value, missing a great amount of joint outliers.

\subsection{Method of consecutive angles}

Visual inspection of a graph like in Figure \ref{fig:Kplots} can be used to determine a reasonable $\Lambda^*$ threshold value for the detection of joint outliers. We propose a method which chooses the threshold value $\Lambda^*$ in a completely unsupervised way, which we refer to as the `method of consecutive angles'. The method relies on the remarks that were made earlier regarding the total number of misclassifications for which $\Lambda_i > \Lambda^*$ and attempts to keep the amount of misclassified joint outliers large, while the number of misclassified inliers remains at low levels. We provide a rough sketch of this method, which uses the rate of decrease as its main tool, instead of the curvature. In order to measure the rate of decrease of the curve, we compute the angle between the line segment joining the number of misclassified points satisfying $\Lambda_i > \Lambda^*$ for consecutive values of $\Lambda^*$ and the horizontal axis. Assuming that $N(\Lambda^*)$ is the number of misclassified points for which $\Lambda_i > \Lambda^*$, the angle between the horizontal axis and the line segment joining $N(\Lambda^*-0.5)$ and $N(\Lambda^*)$, denoted by $\theta_{\Lambda^*}$, is calculated by $\theta_{\Lambda^*} = \arctan\left\{2[N(\Lambda^*-0.5) - N(\Lambda^*) ]\right\}$, for $\Lambda^*= 1.5, 2, \dots, 20$. Our $\Lambda^*$ threshold value is the one for which $\theta_{\Lambda^*} = \theta_{\Lambda^*+0.5}$, meaning that a constant rate of decrease is reached.

Notice that a few considerations need to be made. For instance, the angles only measure the slope but having an equal slope for two consecutive line segments does not necessarily signify a good choice; if the decrease is very large (e.g. we may have 200 misclassified observations for $\Lambda^*$, then $150$ misclassifications for $\Lambda^*+0.5$ and $100$ misclassifications for $\Lambda^*+1$), then it is likely that this choice of $\Lambda^*$ is not ideal. We impose an additional restriction that $\Lambda^*$ is only chosen if $\theta_{\Lambda^*} = \theta_{\Lambda^*+0.5}$ and further $N(\Lambda^*) - N(\Lambda^*+0.5) < \gamma$, where $\gamma$ is a user-specified small integer (we use $\gamma$ for instance). If we are unable to find a $\Lambda^*$ value for which both conditions are satisfied or if that value is unreasonably large (here we consider $\Lambda^*=11$ as the maximum value that is deemed reasonable, in order to allow for some more flexibility), then we go back to the elbow as our $\Lambda^*$ choice. Since the elbow would not be a good choice for the first case, we look at a plot of the angles for different values of $\Lambda^*$ for this data set in Figure \ref{fig:angle_K1}.

\begin{figure}[h!]
  \centering
  \includegraphics[scale=0.5]{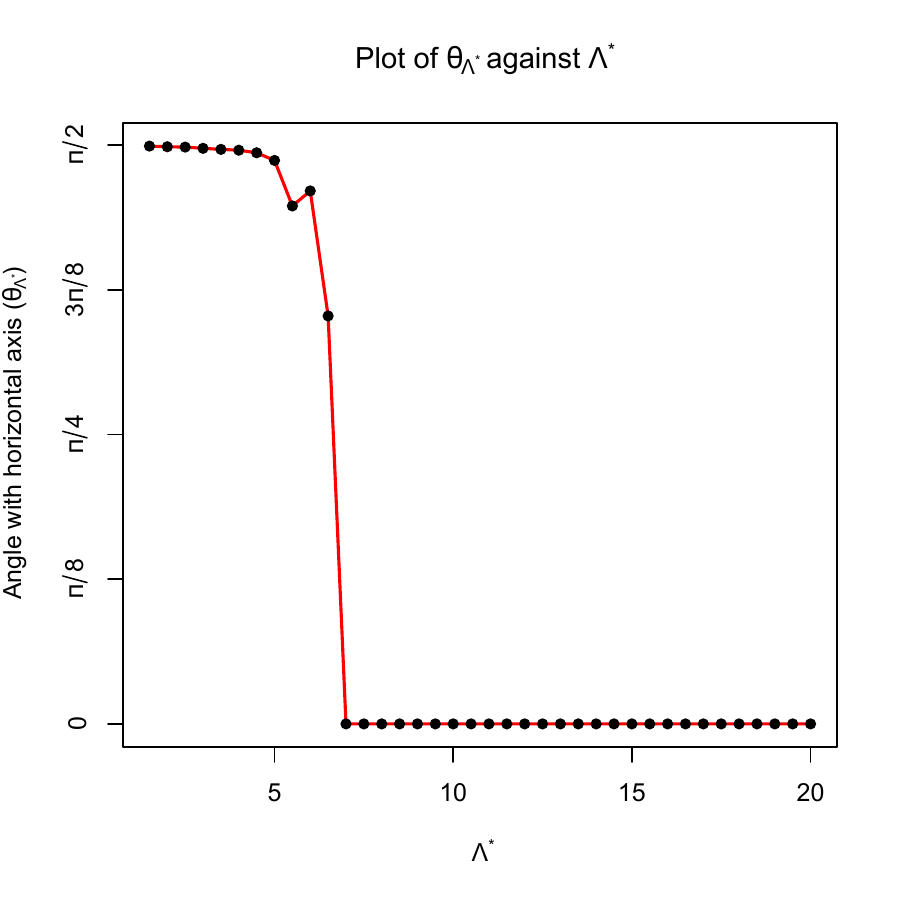}
  \vspace{-0.3cm}
  \caption{Angle with horizontal axis $\theta_{\Lambda^*}$ against values of $\Lambda^*$ for Case 1}
  \label{fig:angle_K1}
\end{figure}

As we can see, we now get a much better value of $\Lambda^*$, equal to 6.5. Using this method for Case 2, we get the exact same result as with the elbow, as a result of the flexibility of our proposed method to choose the elbow point in case it cannot find a value of $\Lambda^*$ satisfying both conditions that we defined earlier. However, when it comes to cases like the third one, the above method would return the $\Lambda^*$ value for which the elbow is detected and this would be a rather poor choice. A much lower value of $\Lambda^*$ would produce much better results in cases like the third one; despite the increase in the number of misclassified inliers, the increase in the number of joint outliers detected would be enough to give us a higher F1 score. We suggest choosing a small value $\Lambda^*_m$ (e.g. $\Lambda^*_m=2$ or 3) for such cases, to allow for some level of overlap of the estimated densities for each class but at the cost of some falsely misclassified inliers. An overview of the method of consecutive angles is outlined in Algorithm \ref{code:consecutive_angles} in Appendix \ref{appendix:algorithms}.

As a matter of fact, the method of consecutive angles is not guaranteed to work perfectly well for all possible settings. For instance, if the proportion of joint outliers is rather large (for instance 20\% of the whole data set) and these are all defined on a binary feature, it is probably safer to choose a much smaller $\Lambda^*$ value of 2 or 3. We recommend proceeding with the method of consecutive angles upon looking at the obtuse angle between the line segments joining the elbow point with the number of misclassified points for $\Lambda^*=1$ and $\Lambda^*=20$. This angle, which we will be referring to as the elbow angle from now on, is calculated by:
\begin{equation}\label{eq:elbowangle}
    \theta_\mathrm{elbow} = \arctan{\left(\left\lvert\frac{\Lambda^*_\mathrm{elbow}-1}{N\left(\Lambda^*_\mathrm{elbow} \right) - N(1)}\right\rvert\right)} + \arctan{\left(\left\lvert\frac{\Lambda^*_\mathrm{elbow}-20}{N\left(\Lambda^*_\mathrm{elbow} \right) - N(20)}\right\rvert\right)},
\end{equation}
where $\Lambda^*_\mathrm{elbow}$ is the value of $\Lambda^*$ where the elbow of the curve is located. A similar expression to \eqref{eq:elbowangle} is mentioned in \cite{zhao2008knee}, who use this to detect the knee of a curve in order to determine the optimal clustering based on the Bayesian Information Criterion (BIC) values obtained for different numbers of clusters. Typically, a small value of $\theta_{\mathrm{elbow}}$ is an indication that the method of consecutive angles is a decent choice, but this also depends on the number of levels of the discrete target feature we are targeting and the dimensionality of the problem. For a binary variable, choosing $\Lambda^*=2$ or 3 is probably the safest option but as the number of levels increases, the method of consecutive angles should be preferred. We have performed a simulation study using the same experimental design as before (Linear, Product, Quotient) for three and four dimensions and varying proportions of joint outliers and number of levels of the target discrete feature. We present the results of this study as guidelines on the choice of method in Table \ref{tab:thetathresh34cont} (Appendix \ref{appendix:extrafigstables}).

We finally look at the F1 scores obtained by selecting the optimal method based on $\theta_\mathrm{elbow}$ for each scenario. It is evident that the performance deteriorates as the number of continuous variables involved increases (see Figure \ref{fig:res_df_F1_Mean}). This comes as no surprise and is a common issue in multivariate kernel density estimation according to \cite{wand1994kernel}, with \cite{scott2015multivariate} having shown that when targeting a univariate and a ten-dimensional standard Gaussian density, the accuracy we get with 50 observations in one dimension can only be achieved with a sample size greater than $10^6$ in the ten-dimensional setting. Therefore, the curse of dimensionality is one of the main limitations of this approach, unless we have a very big number of data points that allows kernel density estimation to produce more reliable results.

\vspace{-0.3cm}
\begin{figure}[h!]
  \centering
  \includegraphics[scale=0.6]{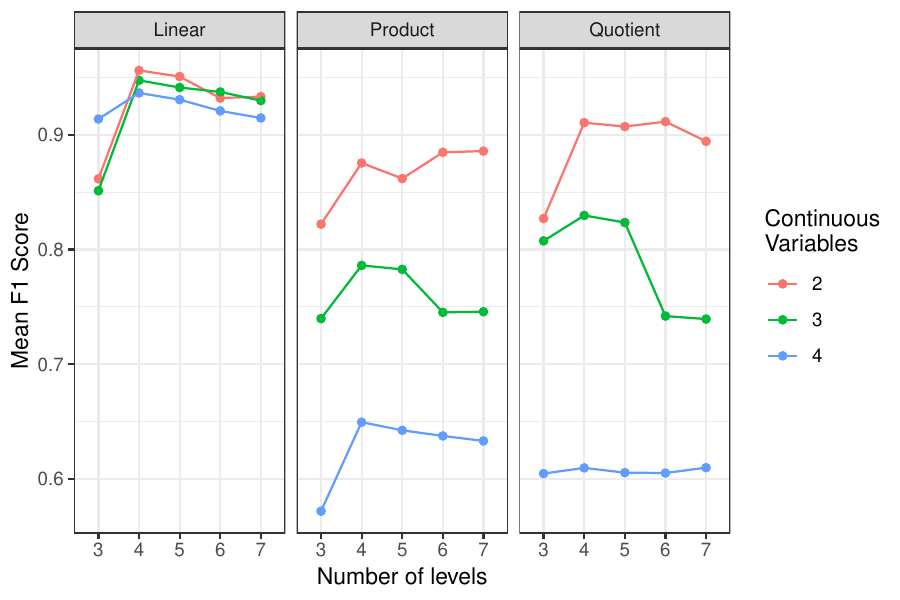}
  \vspace{-0.1cm}
  \caption{Mean F1 Score for linear, product \& quotient designs when there exists a relationship between one discrete with 3--7 levels and 2--4 continuous variables.}
  \label{fig:res_df_F1_Mean}
\end{figure}
Another interesting observation can be made by comparing Figures \ref{fig:res_df_F1_Mean} \& \ref{fig:res_df_Recall_Mean}, which leads to a comparison of the F1 Score to the Recall. Both plots look similar, with similar patterns observed for the same design, number of continuous variables and discrete levels, while the actual values are also close to each other. This is a good indication that the corresponding Precision values are quite high, meaning that the majority of observations flagged are joint outliers. Moreover, it is important to realise that the average Recall is in all but one scenario above 0.5, meaning that over half of the joint outliers are correctly identified. We would also expect these values to increase for more data points, as the curse of dimensionality seems to be the main reason why the Recall (as well as the F1 scores) drops significantly as more continuous variables are introduced. This is much more evident for the product and the quotient designs, where the boundaries of the densities corresponding to different levels are non-linear and should thus need more data points for kernel density estimation to produce more reliable results. 
\vspace{-0.3cm}
\begin{figure}[h!]
  \centering
  \includegraphics[scale=0.6]{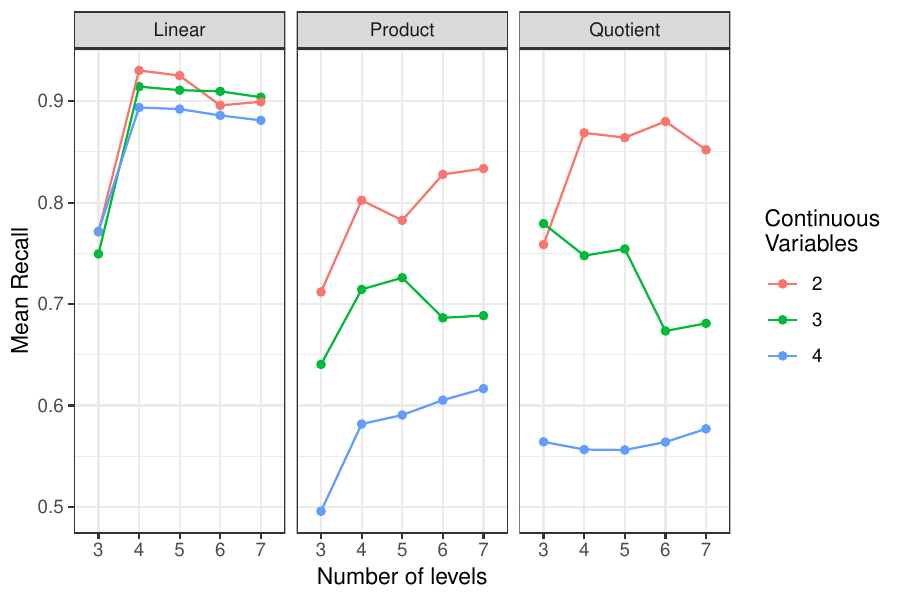}
  \vspace{-0.1cm}
  \caption{Mean Recall for linear, product \& quotient designs when there exists a relationship between one discrete with 3--7 levels and 2--4 continuous variables.}
  \label{fig:res_df_Recall_Mean}
\end{figure}
Notice that the aforementioned phenomenon, also known as the boundary bias problem in the density estimation literature, has been extensively studied and several mitigation proposals have been made. For instance, \cite{silverman1986density} suggests a method for alleviating this issue that is based on reflecting the data about its boundaries. However, this can not be applied in our case as the technique requires knowing the boundaries for each density in advance and these cannot be calculated from the data due to the presence of joint outliers.

\section{Conclusion}

In this paper we present a novel method for detecting outliers in mixed-type data. Our method consists of four steps; calculating scores of outlyingness for the discrete and the continuous features of each observation independently, flagging outliers in any of the two feature spaces, looking for associations between each discrete variable and the set of continuous attributes and detecting outliers in the mixed-attribute space. We assess the performance of our method on data sets with varying characteristics, such as the proportion of outliers within the data, the number of observations, the number of discrete levels of the discrete features and the dimensionality of the space spanned by the continuous variables associated with a discrete feature. Considerably high levels of detection accuracy and a low number of non-outlying observations being falsely flagged as anomalous were achieved.

We defined two types of outliers that a user may encounter within a data set consisting of both discrete and continuous features. While most algorithms only consider observations which are outlying in either just the discrete or the continuous space (these are the marginal outliers, as we have defined them), we have argued that anomalies can also exist in the mixed-feature space and we have defined these to be the joint outliers. The sets of marginal and joint outliers are completely disjoint, with the latter including data points which violate any existing association between a discrete and a set of continuous features.

The calculation of the scores for the discrete variables has been inspired by the ODMAD algorithm of \cite{koufakou2010}. Our revised definition of the score of outlyingness for a discrete feature requires no user-provided threshold values, as these may lead to very inaccurate results; instead, we present a strategy that determines these threshold values in an automated manner, while providing some theoretical results regarding the discrete scores. We further considered and alleviated some of the shortcomings of the definition of the ODMAD score, such as the correlation correction, reducing the computational cost of the process. These slight amendments have led to the discrete scores of the discrete features of the marginal outliers being higher than the respective scores for observations which are not anomalous in the discrete space. As a result, the detection of outliers in the discrete space has been successful in a series of simulations conducted, with very few to none non-outlying observations being detected. However, a comparison to the performance of ODMAD or other less hands-free algorithms for detecting outliers in a mixed-attribute space is infeasible due to the hyperparameter selection step required by such methods, which can yield substantial variations in performance. We finally defined the contribution of each discrete feature to the discrete score.

We used the Extended Isolation Forest algorithm of \cite{hariri2019extended} to compute the scores of the continuous features. In this case, our choice of method has relied on the definition of outliers in the continuous space that is being implied by the algorithm. Having calculated both the discrete and the continuous scores for each observation, we have plotted these against each other on what we call a `score profile' plot. Visual inspection of the score profiles generated for multiple data sets with different characteristics has aided us in developing a method for detecting marginal outliers, which has accomplished an average detection rate of over 99.5\% in our simulation study. Furthermore, the number of falsely detected non-outlying observations has been almost zero for most of the data sets that the method was tested on, with this number starting to grow larger for more than four discrete levels.

The third step of our method deals with identifying associations between discrete and continuous features, upon the removal of the marginal outliers. More precisely, we have devised a procedure that makes use of statistical tests to identify whether a set of continuous features is associated with the levels of a discrete feature. Detecting an association then leads us to the final step of projecting the data in the space spanned by the continuous features identified and using Kernel Density Estimation for predicting the discrete level of the categorical variable for which a relationship was found. We have proposed using a KDE ratio for misclassified observations and have developed a method for calculating a suitable threshold value for this, beyond which any misclassified observations are treated as joint outliers. The strategy deployed has performed very well overall, always detecting over half of the joint outliers on average, albeit a significant amount of non-outlying observations being flagged sometimes as well.

Our proposed method may have managed to deal successfully with the detection of outliers in all three spaces (discrete, continuous \& mixed-attribute), yet it presents certain limitations. As the number of discrete levels gets larger, we have seen that the amount of falsely detected non-outlying observations has increased as well. While in practice one would not expect to encounter a discrete feature with a very large number of levels that frequently, there may be cases of data sets including variables such as the postcode of an individual which can even take thousands of unique values. This is something that may require additional care; for instance, some aggregation may need to be done in advance, so that the number of unique levels decreases. However, this could potentially lead to a loss of valuable information that might be essential for meta-analysis. Additionally, the two-stage test used to detect associations between a discrete and a set of continuous features may fail to discover a structure in instances of multimodal distributions of discrete levels when projected in a set of continuous variables. Finally, although Kernel Density Estimation works and produces rather good results, it is important to be aware of its limitations which we have already mentioned in the relevant section. The curse of dimensionality does play an important role here and if a complex association between a discrete feature and a large number of continuous variables is present, we cannot be certain of a good performance of our method. One may argue that such a relationship (involving a large number of continuous variables) is not so frequent but this is still a possibility that could lead to an abundance of falsely flagged non-outlying observations.

We conclude by giving some recommendations on how the presented method could be used for extending the research that is currently being conducted in certain fields. An interesting application would be that of using our method as a tool for robust cluster analysis of mixed-type data. Enhancing the robustness of clustering algorithms is a challenging problem that has been attempted by multiple researchers \citep[for example, see][]{peel2000robust, browne2015mixture, punzo2016parsimonious, coretto2016robust}, although very few of them have worked with data of more than one type. To the best of our knowledge, there does not exist any non-model-based robust clustering algorithm that can handle data of mixed-type, which highlights the potential for future research on the field. We seek to develop such an algorithm which could possibly improve interpretability of the clustering output. This would require making further adjustments so that the definition of an outlier complies with that of a cluster, which is what has been done in several existing robust clustering algorithms for continuous data, such as Trimmed K-Means \citep{cuesta1997trimmed} and which poses significant challenges as we are working in the mixed-attribute space. It is worth noting that the effect of joint outliers (which are a specific type of contextual outliers in the mixed-attribute domain) in the clustering output has not been explored so far. One could also make use of this method for developing more efficient linear regression techniques; extremities are commonly the reason behind biased parameter estimates and misleading predictions of a regression model. Being able to detect anomalies in advance could help mitigate these problems and some of the ideas presented above could be used in conjunction with concepts used in existing robust regression techniques for the development of even more robust regression algorithms for mixed-type data. These recommendations could pave the way for future research opportunities and culminate in the development of novel methodologies which will in turn provide insight on research questions of interest.

\section*{Software}

The \texttt{R} package \texttt{DOMID} (Detecting Outliers in MIxed-type Data) contains the code and the functions used for the implementation of the method proposed in this paper and we refer the readers to use \texttt{DOMID} for reproducing the simulations described and their findings. The package can be installed through \texttt{GitHub} using \texttt{devtools} by accessing the repository in: \url{https://anonymous.4open.science/r/DOMID-1918/}.

%\backmatter

%\bmhead{Supplementary Material}

%\bmhead{Acknowledgments}
%EC gratefully acknowledges funding provided by EPSRC’s StatML CDT.

\bibliography{main.bib}% common bib file

\begin{thebibliography}{47}
\providecommand{\natexlab}[1]{#1}
\providecommand{\url}[1]{{#1}}
\providecommand{\urlprefix}{URL }
\providecommand{\doi}[1]{\url{https://doi.org/#1}}
\providecommand{\eprint}[2][]{\url{#2}}
 \bibcommenthead

\bibitem[{Aggarwal(2007)}]{aggarwal2007data}
Aggarwal CC (2007) Data Streams: Models and Algorithms, vol~31. Springer

\bibitem[{Agrawal and Srikant(1994)}]{agrawal1994fast}
Agrawal R, Srikant R (1994) Fast algorithms for mining association rules. In:
  Proceedings of the International Conference on Very Large Data Bases, VLDB,
  Citeseer, pp 487--499

\bibitem[{An and Cho(2015)}]{an2015variational}
An J, Cho S (2015) Variational autoencoder based anomaly detection using
  reconstruction probability. Special Lecture on IE 2(1):1--18

\bibitem[{Bouguessa(2015)}]{bouguessa2015practical}
Bouguessa M (2015) A practical outlier detection approach for mixed-attribute
  data. Expert Systems with Applications 42(22):8637--8649

\bibitem[{Browne and McNicholas(2015)}]{browne2015mixture}
Browne RP, McNicholas PD (2015) A mixture of generalized hyperbolic
  distributions. Canadian Journal of Statistics 43(2):176--198

\bibitem[{Chen et~al.(2018)Chen, Yeo, Lee, and Lau}]{chen2018autoencoder}
Chen Z, Yeo CK, Lee BS, et~al. (2018) Autoencoder-based network anomaly
  detection. In: 2018 Wireless Telecommunications Symposium (WTS), IEEE, pp
  1--5

\bibitem[{Coretto and Hennig(2016)}]{coretto2016robust}
Coretto P, Hennig C (2016) Robust improper maximum likelihood: tuning,
  computation, and a comparison with other methods for robust gaussian
  clustering. Journal of the American Statistical Association
  111(516):1648--1659

\bibitem[{Cuesta-Albertos et~al.(1997)Cuesta-Albertos, Gordaliza, and
  Matr{\'a}n}]{cuesta1997trimmed}
Cuesta-Albertos JA, Gordaliza A, Matr{\'a}n C (1997) Trimmed $k$-means: an
  attempt to robustify quantizers. The Annals of Statistics 25(2):553--576

\bibitem[{Di~Pietro and Mancini(2008)}]{di2008intrusion}
Di~Pietro R, Mancini LV (2008) Intrusion Detection Systems, vol~38. Springer
  Science \& Business Media

\bibitem[{Do et~al.(2016)Do, Tran, Phung, and Venkatesh}]{do2016outlier}
Do K, Tran T, Phung D, et~al. (2016) Outlier detection on mixed-type data: An
  energy-based approach. arXiv preprint arXiv:160804830

\bibitem[{Fiore et~al.(2013)Fiore, Palmieri, Castiglione, and
  De~Santis}]{fiore2013network}
Fiore U, Palmieri F, Castiglione A, et~al. (2013) Network anomaly detection
  with the restricted boltzmann machine. Neurocomputing 122:13--23

\bibitem[{Ghoting et~al.(2004)Ghoting, Otey, and
  Parthasarathy}]{ghoting2004loaded}
Ghoting A, Otey ME, Parthasarathy S (2004) {LOADED}: Link-based outlier and
  anomaly detection in evolving data sets. In: Fourth IEEE International
  Conference on Data Mining (ICDM'04), IEEE, pp 387--390

\bibitem[{Han et~al.(2012)Han, Kamber, and Pei}]{datamining2012}
Han J, Kamber M, Pei J (2012) Data Mining Concepts and Techniques, 3rd edn.
  Elsevier

\bibitem[{Hariri et~al.(2019)Hariri, Kind, and Brunner}]{hariri2019extended}
Hariri S, Kind MC, Brunner RJ (2019) Extended isolation forest. IEEE
  Transactions on Knowledge and Data Engineering 33(4):1479--1489

\bibitem[{Holm(1979)}]{holm1979simple}
Holm S (1979) A simple sequentially rejective multiple test procedure.
  Scandinavian Journal of Statistics pp 65--70

\bibitem[{Hubert et~al.(2005)Hubert, Rousseeuw, and
  Vanden~Branden}]{hubert2005robpca}
Hubert M, Rousseeuw PJ, Vanden~Branden K (2005) {ROBPCA}: A new approach to
  robust principal component analysis. Technometrics 47(1):64--79

\bibitem[{Koufakou and Georgiopoulos(2010)}]{koufakou2010}
Koufakou A, Georgiopoulos M (2010) A fast outlier detection strategy for
  distributed high-dimensional data sets with mixed attributes. Data Mining and
  Knowledge Discovery 20(2):259--289

\bibitem[{Koufakou et~al.(2007)Koufakou, Ortiz, Georgiopoulos, Anagnostopoulos,
  and Reynolds}]{koufakou2007scalable}
Koufakou A, Ortiz EG, Georgiopoulos M, et~al. (2007) A scalable and efficient
  outlier detection strategy for categorical data. In: 19th IEEE International
  Conference on Tools with Artificial Intelligence (ICTAI 2007), IEEE, pp
  210--217

\bibitem[{Kruskal and Wallis(1952)}]{kruskal1952use}
Kruskal WH, Wallis WA (1952) Use of ranks in one-criterion variance analysis.
  Journal of the American statistical Association 47(260):583--621

\bibitem[{Liang and Parthasarathy(2016)}]{liang2016robust}
Liang J, Parthasarathy S (2016) Robust contextual outlier detection: Where
  context meets sparsity. In: Proceedings of the 25th ACM International
  Conference on Information and Knowledge Management, pp 2167--2172

\bibitem[{Liu et~al.(2008)Liu, Ting, and Zhou}]{liu2008isolation}
Liu FT, Ting KM, Zhou ZH (2008) Isolation forest. In: 2008 Eighth IEEE
  International Conference on Data Mining, IEEE, pp 413--422

\bibitem[{Loader(1999)}]{loader2006local}
Loader C (1999) Local Regression and Likelihood. Springer

\bibitem[{Meghanath et~al.(2018)Meghanath, Pai, and
  Akoglu}]{meghanath2018conout}
Meghanath M, Pai D, Akoglu L (2018) Con{O}ut: Contextual outlier detection with
  multiple contexts: Application to ad fraud. In: Joint European Conference on
  Machine Learning and Knowledge Discovery in Databases, Springer, pp 139--156

\bibitem[{Ngai et~al.(2011)Ngai, Hu, Wong, Chen, and Sun}]{ngai2011application}
Ngai EW, Hu Y, Wong YH, et~al. (2011) The application of data mining techniques
  in financial fraud detection: A classification framework and an academic
  review of literature. Decision Support Systems 50(3):559--569

\bibitem[{Otey et~al.(2006)Otey, Ghoting, and Parthasarathy}]{otey2006fast}
Otey ME, Ghoting A, Parthasarathy S (2006) Fast distributed outlier detection
  in mixed-attribute data sets. Data Mining and Knowledge Discovery
  12(2):203--228

\bibitem[{Pang et~al.(2021)Pang, Shen, Cao, and Hengel}]{pang2021deep}
Pang G, Shen C, Cao L, et~al. (2021) Deep learning for anomaly detection: A
  review. ACM Computing Surveys (CSUR) 54(2):1--38

\bibitem[{Parzen(1962)}]{parzen1962estimation}
Parzen E (1962) On estimation of a probability density function and mode. The
  Annals of Mathematical Statistics 33(3):1065--1076

\bibitem[{Peel and McLachlan(2000)}]{peel2000robust}
Peel D, McLachlan GJ (2000) Robust mixture modelling using the t distribution.
  Statistics and Computing 10(4):339--348

\bibitem[{Preiss(1999)}]{bruno1999data}
Preiss BR (1999) Data Structure and Algorithms with Object-Oriented Design
  Patterns in Java. John Wiley \& Sons

\bibitem[{Punzo and McNicholas(2016)}]{punzo2016parsimonious}
Punzo A, McNicholas PD (2016) Parsimonious mixtures of multivariate
  contaminated normal distributions. Biometrical Journal 58(6):1506--1537

\bibitem[{Satop{\"a}{\"a} et~al.(2011)Satop{\"a}{\"a}, Albrecht, Irwin, and
  Raghavan}]{satopaa2011finding}
Satop{\"a}{\"a} V, Albrecht J, Irwin D, et~al. (2011) Finding a ``kneedle'' in
  a haystack: Detecting knee points in system behavior. In: 31st International
  Conference on Distributed Computing Systems Workshops, IEEE, pp 166--171

\bibitem[{Savage et~al.(2014)Savage, Zhang, Yu, Chou, and
  Wang}]{savage2014anomaly}
Savage D, Zhang X, Yu X, et~al. (2014) Anomaly detection in online social
  networks. Social Networks 39:62--70

\bibitem[{Schlegl et~al.(2017)Schlegl, Seeb{\"o}ck, Waldstein, Schmidt-Erfurth,
  and Langs}]{schlegl2017unsupervised}
Schlegl T, Seeb{\"o}ck P, Waldstein SM, et~al. (2017) Unsupervised anomaly
  detection with generative adversarial networks to guide marker discovery. In:
  Proceedings of the 25th International Conference on Information Processing in
  Medical Imaging, Springer, pp 146--157

\bibitem[{Scott(2015)}]{scott2015multivariate}
Scott DW (2015) Multivariate Density Estimation: Theory, Practice, and
  Visualization, John Wiley \& Sons, chap~7, pp 195--217

\bibitem[{Shannon(1948)}]{shannon1948mathematical}
Shannon CE (1948) A mathematical theory of communication. The Bell System
  Technical Journal 27(3):379--423

\bibitem[{Silverman(1986)}]{silverman1986density}
Silverman BW (1986) Density Estimation for Statistics and Data Analysis,
  vol~26. Chapman and Hall

\bibitem[{Sison and Glaz(1995)}]{sison1995simultaneous}
Sison CP, Glaz J (1995) Simultaneous confidence intervals and sample size
  determination for multinomial proportions. Journal of the American
  Statistical Association 90(429):366--369

\bibitem[{Tch\'ebychef(1867)}]{chebyshev}
Tch\'ebychef PL (1867) Des valeurs moyennes. Journal de Math\'ematiques Pures
  et Appliqu\'ees

\bibitem[{Theil(1970)}]{theil1970estimation}
Theil H (1970) On the estimation of relationships involving qualitative
  variables. American Journal of Sociology 76(1):103--154

\bibitem[{Thudumu et~al.(2020)Thudumu, Branch, Jin, and
  Singh}]{thudumu2020comprehensive}
Thudumu S, Branch P, Jin J, et~al. (2020) A comprehensive survey of anomaly
  detection techniques for high dimensional big data. Journal of Big Data
  7:1--30

\bibitem[{Tran et~al.(2011)Tran, Phung, and Venkatesh}]{tran2011mixed}
Tran T, Phung D, Venkatesh S (2011) Mixed-variate restricted boltzmann
  machines. In: 3rd Asian Conference on Machine Learning, PMLR, pp 213--229

\bibitem[{Tschuchnig and Gadermayr(2022)}]{tschuchnig2022anomaly}
Tschuchnig ME, Gadermayr M (2022) Anomaly detection in medical imaging - a mini
  review. In: Data Science -- Analytics and Applications. Springer, pp 33--38

\bibitem[{Wand and Jones(1994)}]{wand1994kernel}
Wand MP, Jones MC (1994) Kernel Smoothing. CRC press

\bibitem[{Zaykin et~al.(2002)Zaykin, Zhivotovsky, Westfall, and
  Weir}]{zaykin2002truncated}
Zaykin DV, Zhivotovsky LA, Westfall PH, et~al. (2002) Truncated product method
  for combining {$P$}-values. Genetic Epidemiology 22(2):170--185

\bibitem[{Zhang et~al.(2016)Zhang, Yu, Li, Zhang, Xun, and
  Qin}]{zhang2016relevant}
Zhang J, Yu X, Li Y, et~al. (2016) A relevant subspace based contextual outlier
  mining algorithm. Knowledge-Based Systems 99:1--9

\bibitem[{Zhang and Jin(2011)}]{zhang2011effective}
Zhang K, Jin H (2011) An effective pattern based outlier detection approach for
  mixed attribute data. In: AI 2010: Advances in Artificial Intelligence: 23rd
  Australasian Joint Conference, Adelaide, Australia, December 7-10, 2010.
  Proceedings 23, Springer, pp 122--131

\bibitem[{Zhao et~al.(2008)Zhao, Hautamaki, and Fr{\"a}nti}]{zhao2008knee}
Zhao Q, Hautamaki V, Fr{\"a}nti P (2008) Knee point detection in bic for
  detecting the number of clusters. In: Advanced Concepts for Intelligent
  Vision Systems: 10th International Conference, ACIVS 2008, Proceedings 10,
  Springer, pp 664--673

\end{thebibliography}
\newpage
\begin{appendices}

\section{Experimental Design}\label{appendix:design}

For a data set of $n$ observations and $p$ features, out of which $p_D$ are discrete and $p_C$ are continuous, we generate $n$ observations from a $p$-variate Gaussian distribution. The mean is set to be equal to a $p$-dimensional vector of zeros. The covariance matrix is a positive definite $(p \times p)$-dimensional matrix with random diagonal values (corresponding to the component variances) restricted in the range $[0.1, 5]$. The \texttt{genPositiveDefMat} function from the \texttt{clusterGeneration} package is used for this purpose and we set \texttt{ndim = }$p$, \texttt{covMethod = "unifcorrmat"}, \texttt{alphad = 5} and \texttt{rangeVar = c(0.1, 5)}. Discrete features are then obtained by discretising $p_D$ of the $p$ variables using quantile discretisation and by providing a number of discrete levels for each discrete feature. We explain how we generate marginal and joint outliers in the subsections below.

\subsection*{Marginal Outliers}

We denote the proportion of outliers in the data set by $q$, while $q_M$ and $q_J$ represent the proportions of marginal and joint outliers respectively. We randomly choose $n \times q_M$ observations, half of which will be outlying in the discrete space and the remaining will be outlying in the continuous domain. For the discrete outliers, we randomly sample an integer $z$ from $\{1, \dots, p_D\}$ for each of the $n \times q_M /2$ observations selected and choose $z$ discrete variables. We then change the level of these $z$ discrete features of the observation and set it equal to $\ell_j + 1$ for each discrete variable $j$ with $\ell_j$ levels. As an example, if $z=1$ for an observation $i$ and the binary feature $j$ is chosen, we set $x_{i,j} = 3$.

Continuous outliers are simulated in a similar manner to the discrete outliers. We once again choose an integer $z$, this time from $\{1, \dots, p_C\}$ and pick $z$ continuous features in random. We add or subtract fifteen units from the observation values for each of the continuous variables selected. The choice of fifteen is such as to ensure that even in the most extreme case of a maximum allowed variance of five units, adding or subtracting fifteen units can guarantee that the point is still far enough from the rest along each continuous component. Combined marginal outliers are finally generated by drawing an integer from $\{ 1, \dots , n\times q_M/2\}$, selecting this amount of discrete outliers and repeating the process for generating continuous outliers on them, so that they are outlying in both the discrete and the continuous domains.

\subsection*{Joint Outliers}

In order to generate $n \times q_J$ joint outliers, we first need to define associations between discrete and sets of continuous features. Assuming without loss of generality that we wish to impose a linear association between the first discrete and the first two continuous features, we start by looking at the quantiles of the differences between the values of the first two continuous variables. We use quantile discretisation with $\ell_1$ levels to produce new levels for the first discrete variable, where $\ell_1$ refers to the number of discrete levels of the first discrete feature (prior to the generation of marginal outliers). The observations that had been chosen to be marginally outlying in the discrete or in both domains remain unaffected. Then we randomly select $n \times q_J$ observations among the ones which are not marginally outlying and we change the level of their first discrete variable, ensuring that this is still a value from $\{ 1, \dots, \ell_1\}$. If we wish to generate $n_A$ associations, the process is repeated $n_A$ times (with the levels of the target discrete features being determined by the type of association we wish to enforce) and the number of observations for which the discrete level value is altered for each association is given by $\lfloor n \times q_J / n_A \rfloor$.

\section{Proofs}\label{appendix:proofs}
\setcounter{figure}{0}
\renewcommand\thefigure{B.\arabic{figure}}
\setcounter{table}{0}
\renewcommand\thetable{B.\arabic{table}}
\setcounter{equation}{0}
\renewcommand\theequation{B\arabic{equation}.\arabic{prop}}
\begin{prop}\label{prop:avgdiscscore}
The average of all discrete scores $m_D$ satisfies:$$
\frac{1}{n\left(\Xi-1\right)\left(\maxlen\right)^2} \leq m_D \leq \frac{p_D\left(\lvert \mathcal{E}\rvert -1\right)}{n},$$
where $\mathcal{E}$ is the set of all unique values of the discrete scores $s_{D,(i,\cdot)} \ \left(i\in\{1,\dots,n\}\right)$, $\lvert \mathcal{E}\rvert$ is the cardinality of $\mathcal{E}$ which is assumed to be greater than one and $\Xi$ is the maximum of all threshold values $\sigma_d$ for itemsets $d$ with $\lvert d \rvert = \maxlen$ (i.e. $\Xi=\max_{d: \lvert d \rvert = \maxlen}\left\{\sigma_d\right\}$).
\end{prop}
\begin{proof}
    Let $S_X$ be the random variable that is associated with the discrete scores for a data set $\boldsymbol{X}$ consisting of $p_D$ discrete features and denote the realisations of $S_X$ by $s_X$. We also define $\mathcal{E}$ as the set that contains all unique discrete score values and denote its cardinality by $\lvert \mathcal{E} \rvert$. We clearly have that $0 \in \mathcal{E}$ and we assume that $\lvert \mathcal{E} \rvert > 1$, meaning that at least one discrete outlier exists or that even if no outliers are present in the data, then there exists at least one observation in $\boldsymbol{X}$ with non-zero discrete score. The probability mass function of $S_X$ is given by:$$
    p_{S_X}\left(s_X\right) = q\left(s_X\right)\mathbbm{1}\left\{s_X \in \mathcal{E}\right\},$$
    where $\mathbbm{1}\{\cdot\}$ is an indicator function and $q\left(s_X\right)$ is defined as the proportion of observations in $\boldsymbol{X}$ with discrete score $s_X$. More precisely:$$
    q\left(s_X\right) = \frac{n_{s_X}}{n} = \frac{\# \ \mathrm{of \ observations \ in \ }\boldsymbol{X} \ \mathrm{with \ discrete \ score \ }s_X}{\# \ \mathrm{of \ observations \ in \ }\boldsymbol{X}}.$$
    We now recall the definition of the discrete score as given in Expression \eqref{eq:discretescorefinal}:$$
    s_{D,(i,\cdot)}=\sum_{\substack{d \subseteq \boldsymbol{X}_{i,D}: \\ \\\operatorname{supp}(d)<\sigma_d, \\ |d| \leq \maxlen, \\ \left\{\left\{k, k'\right\}: u_{k,k'}>u^{\mathrm{upper}}_{k,k'}\right\}\nsubseteq d}} \frac{1}{\operatorname{supp}(d) \times\lvert d \rvert^2}, \quad i=1,\dots,n.$$
    Hence, we have that for an observation $i$ with discrete score $s_X$, the product $s_Xq\left(s_X\right)$ is equal to:$$
    \sum_{\substack{d \subseteq \boldsymbol{X}_{i,D}: \\ \\\operatorname{supp}(d)<\sigma_d, \\ |d| \leq \maxlen, \\ \left\{\left\{k, k'\right\}: u_{k,k'}>u^{\mathrm{upper}}_{k,k'}\right\}\nsubseteq d}} \frac{1}{\operatorname{supp}(d) \times\lvert d \rvert^2}\times \frac{n_{s_X}}{n}.$$
    We observe that $\operatorname{supp}(d)$ is the number of times each itemset $d$ included in $\boldsymbol{X}_{i,D}$ appears within $\boldsymbol{X}$. This is bounded above by $n_{s_X}/r$, where $1 \leq r \leq p_D$. Essentially, the support of an itemset $d$ is the one that yields the score $s_X$, while this could appear in more than just one discrete variable. As an illustrative example, we can consider, without loss of generality, the case of the first $p_D$ observations containing a unique discrete level for discrete features $1, \dots, p_D$, respectively and assume that no other infrequent itemsets are included in these observations. In such a case, the support of each infrequent itemset $d$ with $\lvert d \rvert = 1$ will be equal to a unit and the score for each of these $p_D$ observations will be equal to a unit as well, yielding $s_X=1$, $n_{s_X}=p_D$ and $r=p_D$. Notice that in fact, the value of $r$ could be greater than $p_D$; $r$ represents the number of unique combinations of features which produce the score $s_X$. However, if we were to consider itemsets of length two with unit support, it is extremely unlikely that we would observe these on all $C^{p_D}_{2}$ possible combinations of discrete variables and it is also very unlikely that such observations exist, without containing any infrequent discrete levels (which would lead to pruning). Moreover, the number of outliers in $\boldsymbol{X}$ is assumed to be small enough for us to safely assume that observations with infrequent itemsets of unit length are dominating in terms of the amount of observations with non-zero discrete scores. Thus, we can calculate the expectation of $S_X$ as:$$
    m_D = \mathbb{E}\left(S_X\right) = \sum\limits_{s_X \in \mathcal{E}} s_Xq\left(s_X\right) = \sum\limits_{s_X \in \mathcal{E}\backslash\{0\}} s_Xq\left(s_X\right)$$
    where using the fact that $\operatorname{supp}(d)=n_{s_X}/n$ and that $\lvert d \rvert \geq 1$, we get that $s_Xq\left(s_X\right)$ is bounded above by $p_D/n$. The lowest possible score, assuming there exists at least one observation with a non-zero discrete score, is attained for an infrequent itemset of length equal to $\maxlen$ which occurs one time less than the maximum threshold value out of all threshold values for itemsets of length $\maxlen$. Therefore, we obtain:$$
     \frac{1}{n\left(\maxlen\right)^2\left(\Xi-1\right)}\leq m_D \leq \frac{p_D\left(\lvert\mathcal{E}\rvert-1\right)}{n},
    $$
which is the required result.
\end{proof}
\setcounter{equation}{0}
\renewcommand\theequation{B\arabic{equation}.\arabic{prop}}
\begin{prop}\label{prop:sddiscscore}
    The sample standard deviation of the discrete scores, denoted by $s_D$, is bounded below by $L_B$ and bounded above by $U_B$, defined as:$$\begin{aligned}
    L_B & = \frac{1}{\sqrt{n}\left(\maxlen\right)^2\left(\Xi-1\right)},\\
    U_B & = \frac{1}{\left(\maxlen\right)^2\left(\Xi-1\right)}\sqrt{\frac{p_D\left(\lvert\mathcal{E}\rvert-1\right)\left(\maxlen\right)^4s_X^{\left(\lvert\mathcal{E}\rvert\right)}\left(\Xi-1\right)^2-1}{n(n-1)}},\end{aligned}$$
    where $s_X^{\left(\lvert\mathcal{E}\rvert\right)}$ is the maximum discrete score for a data set $\boldsymbol{X}$ with $p_D$ discrete features, $\mathcal{E}$ is the set of all unique values of the discrete scores $s_{D,(i,\cdot)} \ \left(i\in\{1,\dots,n\}\right)$, $\lvert \mathcal{E}\rvert$ is the cardinality of $\mathcal{E}$ which is assumed to be greater than one and $\Xi$ is the maximum of all threshold values $\sigma_d$ for itemsets $d$ with $\lvert d \rvert = \maxlen$ (i.e. $\Xi=\max_{d: \lvert d \rvert = \maxlen}\left\{\sigma_d\right\}$).
\end{prop}

\begin{proof}
    We first notice that the sample standard deviation for the discrete scores is given by:$$
    s_D = \sqrt{\frac{\sum\limits_{s_X \in \mathcal{E}}s_X^2q\left(s_X\right) - nm_D^2}{n-1}},$$
    and we know from Proposition \eqref{prop:avgdiscscore} that $m_D \geq 1/\left(n\left(\maxlen\right)^2\left(\Xi-1\right)\right)$. Moreover, from the proof of the aforemenetioned Proposition, we have that:$$
    \sum\limits_{s_X \in \mathcal{E}}s_X^2q\left(s_X\right) \leq \frac{p_D}{n}\sum\limits_{s_X \in \mathcal{E}}s_X \leq \frac{p_D\left(\lvert\mathcal{E}\rvert-1\right)}{n}s_X^{\left(\lvert\mathcal{E}\rvert\right)},$$
    where we define $s_X^{\left(\lvert\mathcal{E}\rvert\right)}$ to be the maximum discrete score observed. The upper bound follows directly by substituting the upper and lower bounds mentioned for the two terms.\\

    \noindent For the lower bound, since we assume $\lvert \mathcal{E}\rvert>1$, the minimum standard deviation cannot be zero (which would've been the case had $\lvert \mathcal{E}\rvert = 1$) but its minimum value is attained if just one observation has the minimum possible non-zero discrete score. That is the case of an infrequent itemset of length $\maxlen$ that appears one time less than the maximum threshold value our of all thresholds for itemsets of length $\maxlen$, which we defined by $\Xi$. In such a case, the standard deviation can be shown to be equal to the expression for the lower bound. 
\end{proof}
\setcounter{equation}{0}
\renewcommand\theequation{B.\arabic{prop}.\arabic{equation}}
\begin{prop}\label{prop:maxdiscscore_expression}
    The maximum value of the discrete score $s_{D,(i, \cdot)}$ for the $i$th observation of a data set with $p_D$ discrete variables is attained for all itemsets of length $k \leq p_D$ appearing just once.
\end{prop}
\begin{proof}
    We assume that we have $p_D \geq 2$ discrete variables (for the case of $p_D=1$, the result follows immediately by the definition of the discrete score) and recall the formulation of the discrete score for an observation $i$ as given in Expression \eqref{eq:discretescorefinal}:$$
    s_{D,(i,\cdot)}=\sum_{\substack{d \subseteq \boldsymbol{X}_{i,D}: \\ \\\operatorname{supp}(d)<\sigma_d, \\ |d| \leq \maxlen, \\ \left\{\left\{k, k'\right\}: u_{k,k'}>u^{\mathrm{upper}}_{k,k'}\right\}\nsubseteq d}} \frac{1}{\operatorname{supp}(d) \times\lvert d \rvert^2}, \quad i=1,\dots,n.$$
    Assuming infinitesimal nominal association between the discrete variables and a large $\maxlen$ value (to avoid having many restrictions on the calculation of the score), it is easy to see that itemsets of unit support yield a higher increase in the score. Therefore, we need all itemsets contributing to the score to appear just once in the data set. In this case, the total score becomes equal to the following expression:$$
    \mathcal{A} = p_D - \sum\limits_{i=1}^{p_D-1}\alpha_i + \sum\limits_{i=1}^{p_D-1}\dbinom{\alpha_i}{i+1}\frac{1}{\left(i+1\right)^2}, \quad \alpha_i \in \left\{0, i+1, \dots, p_D\right\}, \ \sum\limits_{i=1}^{p_D-1}\alpha_i \leq p_D.$$
    Notice that the first two terms correspond to the contribution of the itemsets of unit length that appear once in the data set, while the third term corresponds to the contribution of all possible itemsets of greater length (up to length $p_D$) to the score. The coefficients $\alpha_i$ represent the number of discrete variables which are included in infrequent itemsets (of unit support) of length $i+1$. We further impose the restriction that the sum of the $\alpha_i$'s is at most equal to $p_D$, since any itemsets of length $i+1$ are defined based on $i+1$ variables, and we are restricted to $p_D$ discrete features. Moreover, we could potentially have no infrequent itemsets of length $i+1$ (in which case $\alpha_i=0$) but if we do have any, then these should be observed in at least $i+1$ discrete features. Despite the abuse of notation, we assume that $C^0_{i+1}=0$ for convenience.\\
    
    \noindent Our goal is to find the set of values $\boldsymbol{\alpha}=\left(\alpha_1, \dots, \alpha_{p_D-1} \right)$ that maximise $\mathcal{A}$ and more precisely, we aim to show that this is achieved either when all the $\alpha_i$'s are equal to zero or when we are on the boundary of the solution space, i.e. when all the $\alpha_i$'s sum to $p_D$ and all of them but one are exactly zero. The first case corresponds to maximising the discrete score for all $p_D$ discrete features of an observation being unique within the data set, while the second is interpreted as achieving the maximum discrete score possible when all itemsets of length $i+1$ that can be generated from $p_D$ discrete variables appear just once and there exist no itemsets of greater or lower length which are infrequent. In both these cases, we basically end up with the conclusion that the maximum score is attained for all itemsets of one specific length occurring once in the data set. This means that for instance, we cannot maximise the discrete score of an observation for a combination of $\alpha_1$ itemsets of length two and $\alpha_2$ itemsets of length three appearing once in the data (where $\alpha_1, \alpha_2 > 0$ and $\alpha_1+\alpha_2 \leq p_D$), assuming we have at least $p_D \geq 5$ discrete features.\\

    \noindent We begin by showing that for $p_D \leq 8$, $\mathcal{A}$ is maximised for all $\alpha_i$'s being equal to 0. We would rather have $\alpha_i=0 \ \forall i$ if:$$
    \sum\limits_{i=1}^{p_D-1}\alpha_i > \sum\limits_{i=1}^{p_D-1}\dbinom{\alpha_i}{i+1}\frac{1}{(i+1)^2}.$$
    The above condition can be equivalently written, by some straightforward algebraic manipulations, as:$$
    \sum\limits_{i=1}^{p_D-1}\left\{ \alpha_i \times \left( \frac{\left(\alpha_i-1\right) \times \dots \times \left(\alpha_i-i\right)}{(i+1)^3\times i!}-1\right)\right\} < 0.$$
    Notice how since we assume that all the $\alpha_i$'s are non-negative, the expression above can be negative if and only if there exists at least a non-zero $\alpha_i$ value such that:
    \begin{equation}\label{eq:aiszero}
    F = \frac{\left(\alpha_i-1\right) \times \dots \times \left(\alpha_i-i\right)}{(i+1)^3\times i!}-1 = \frac{\Gamma\left(\alpha_i\right)}{\Gamma\left(\alpha_i-i\right)\times\Gamma\left(i+1\right)\times (i+1)^3} - 1 < 0,
    \end{equation}
    where $\Gamma(\cdot)$ is the gamma function. Now if the above expression is always negative for all possible non-zero values of the $\alpha_i$'s, we can ensure that $\mathcal{A}$ is maximised for all the $\alpha_i$'s being equal to zero. We plot the value of $F$ for all possible combinations of $i$ and non-zero $\alpha_i$ values in Figure \ref{fig:appendixproof1pltais}, for $p_D \in \{7, 8, 9, 10\}$. We can see that for $p_D \leq 8$, the value of $F$ is always negative, while for $p_D \geq 9$ some combinations of $i$ and $\alpha_i$ values yield non-negative $F$ values, which contradict the fact that $\alpha_i=0 \ \forall i$ is the optimal solution. This is also the case for larger $p_D$ values. Therefore, $\alpha_i = 0 \ \forall i$ maximises expression $\mathcal{A}$ for $p_D \leq 8$.\\
   
\begin{figure}[h!]
    \centering
    \includegraphics[scale=0.7]{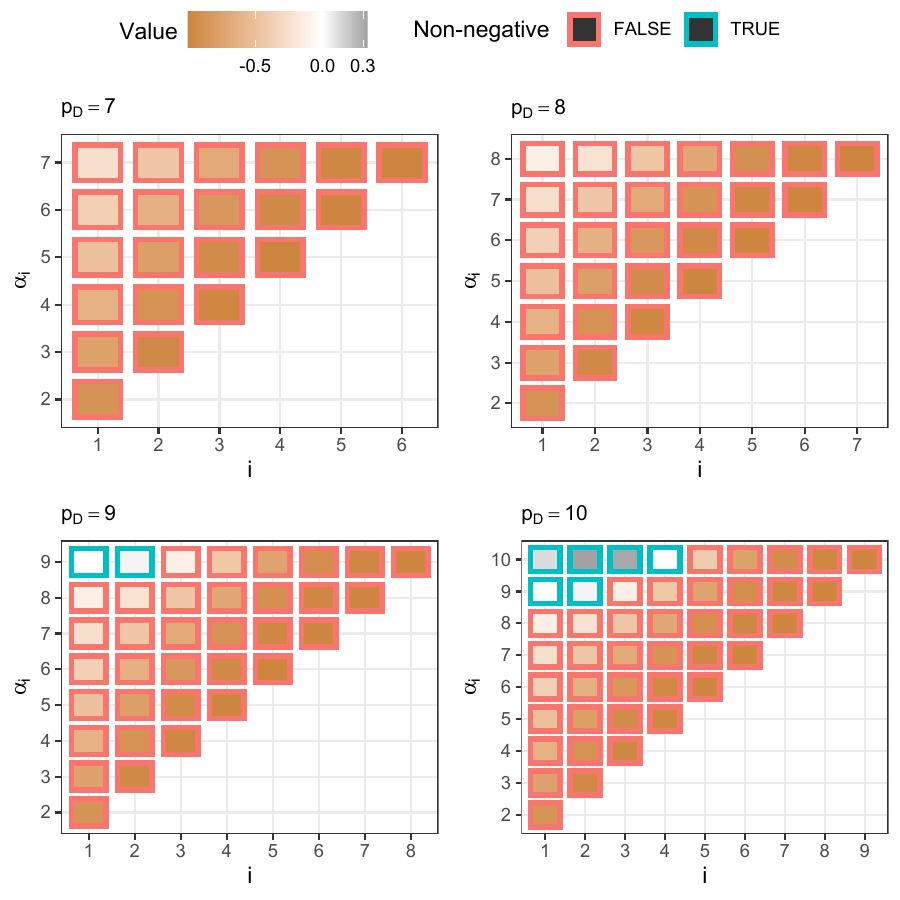}
    \vspace{-0.1cm}
    \caption{Value of $F$ (Expression \eqref{eq:aiszero}) for all possible combinations of $i$ and non-zero $\alpha_i$ values for $p_D \in \{7, 8, 9, 10\}$. The fill color of each box represents the value of $F$ and the border color denotes whether that value is non-negative. For $p_D \leq 8$ all possible combinations lead to a negative value of $F$, while for $p_D \geq 9$, some non-negative values are observed too.}
    \label{fig:appendixproof1pltais}
\end{figure}

\noindent From this point onwards, we will be assuming that $p_D \geq 9$. Our strategy consists of first showing that getting to the solution boundary is always preferable when having one non-zero $\alpha_i$ and then, we will show that setting $\alpha_j$ equal to a non-zero value, where $i \neq j$, decreases the value of $\mathcal{A}$. We begin by assuming that there exists an index $j$ such that $\alpha_j > 0$ and $\alpha_i=0 \ \forall i \neq j$ and we look at what happens to the value of $\mathcal{A}$ if we decrease $\alpha_j$ by a unit. Notice that due to the restriction on the values that $\alpha_j$ can take, we need to assume that $j \neq p_D -1$, since $\alpha_{p_D-1}$ can only be equal to 0 or $p_D-1$. We denote the loss that results from decreasing $j$ by a unit by $G_1$ and that is equal to:
\begin{align*}
    G_1 & = p_D - \alpha_j + \dbinom{\alpha_j}{j+1}\times \frac{1}{(j+1)^2} - \left\{ p_D - \alpha_j + 1 + \dbinom{\alpha_j-1}{j+1}\times \frac{1}{(j+1)^2}\right\}\\
    & = \frac{1}{(j+1)^2}\left\{ \dbinom{\alpha_j}{j+1} - \dbinom{\alpha_j-1}{j+1} \right\}-1\\
    & = \frac{\Gamma\left(\alpha_j\right)}{(j+1)\times \Gamma\left(j+2\right)\times \Gamma\left( \alpha_j-j\right)} - 1.
\end{align*}
We can also look at the case of decreasing $\alpha_j$ from its minimum possible non-zero value of $j+1$ to zero. We define $G_2$ to be the loss corresponding to this decrease:
\begin{align*}
    G_2 & = p_D - \alpha_j + \dbinom{\alpha_j}{j+1}\times \frac{1}{(j+1)^2} - p_D\\
    & = -(j+1) + \dbinom{j+1}{j+1}\times\frac{1}{(j+1)^2}\\
    & = \frac{1}{(j+1)^2}-(j+1).
\end{align*}
It is rather evident that for $j \in \mathbb{Z}^+$, the value of $G_2$ will be negative. A negative loss is equivalent to a gain, therefore we get a higher $\mathcal{A}$ value when going from $\alpha_j=j+1$ to $\alpha_j=0$. Deriving such a result for $G_1$ is not as straightforward but we can perform some simulations to see what the behaviour of $G_1$ is for different combinations of $j$ and $\alpha_j$ values as $p_D$ varies. This is illustrated in Figure \ref{fig:loss_decrease}, where the off-diagonal elements correspond to the value of $G_1$, while the diagonal ones correspond to the value of $G_2$. These are the loss values when decreasing $\alpha_j$ by a unit and when setting $\alpha_j$ from $j+1$ to zero, respectively.\\

\noindent We can observe a very specific pattern in the loss values; more precisely, the loss is negative if $\alpha_j \leq j+3$, which means that decreasing $\alpha_j$ will lead to a greater $\mathcal{A}$ value. However, if $\alpha_j \geq j+4$, decreasing $\alpha_j$ is apparently not a good idea, as that would yield a lower $\mathcal{A}$ value, the loss being larger as $\alpha_j$ increases. Notice that none of these loss values is equal to zero and this can be easily seen by setting $\alpha_j = j+3$ and substituting this into the expression for $G_1$ to derive an asymptotic bound for the minimum absolute value of the loss. The exact same pattern can be observed for greater $p_D$ values than the ones we have included here, but these are omitted. The above suggest that we should either have every single $\alpha_j$ value equal to zero or that we should be on the boundary, where $\alpha_j=p_D$. We have already seen that setting all parameters equal to zero is not desirable for $p_D \geq 9$, thus we can conclude that the optimal solution is on the boundary for this range of $p_D$ values.\\

\begin{figure}[h!]
    \centering
    \includegraphics[scale=0.7]{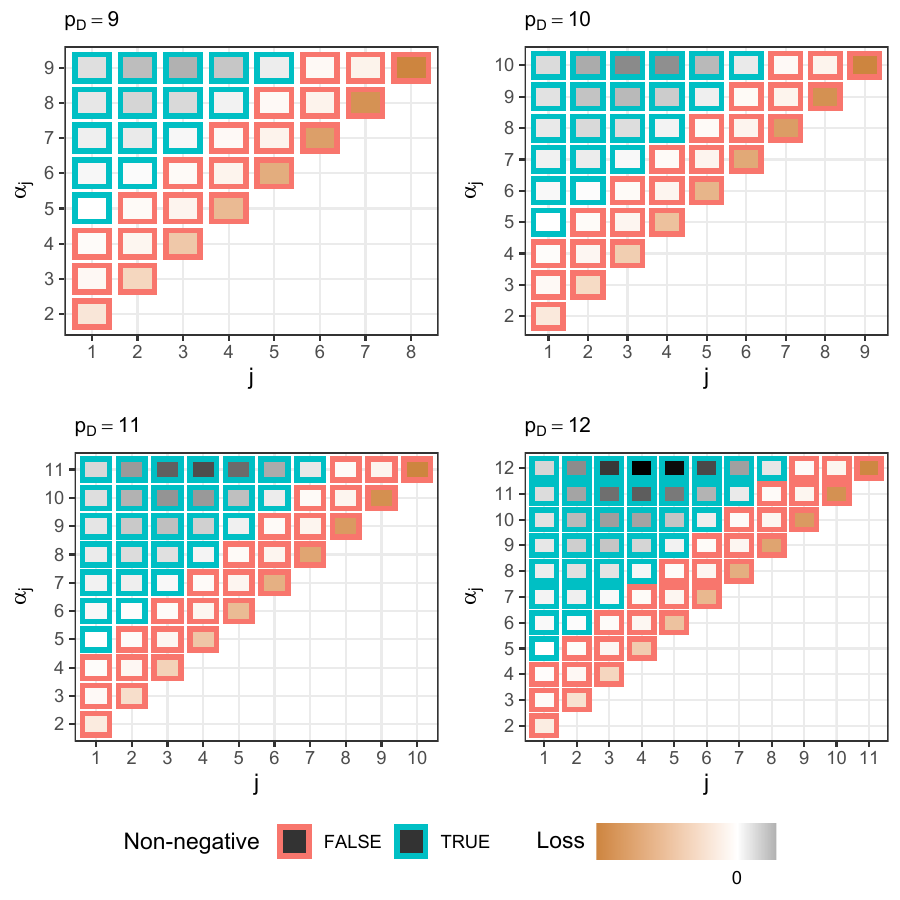}
    \vspace{-0.1cm}
    \caption{Value of $G_1$ (off-diagonal elements) for all possible combinations of $j$ and $\alpha_j$ and of $G_2$ (diagonal elements) for $p_D \in \{9, 10, 11, 12\}$. The fill color of each box represents the value of the respective loss incurred when decreasing $\alpha_j$ by a unit (off-diagonal elements) or when setting $\alpha_j$ from its minimum non-zero value of $j+1$ to zero (diagonal elements) and the border color denotes whether that value is non-negative.}
    \label{fig:loss_decrease}
\end{figure}

\noindent So far, we have managed to show that if $p_D \leq 8$, the optimal solution is for all the $\alpha_i$'s being zero, while for $p_D \geq 9$, if one of the $\alpha_i$'s is non-zero, it should be equal to $p_D$ so that $\mathcal{A}$ is maximised. The final part of the proof consists of showing that if we are on the boundary of the solution space, with $\alpha_j=p_D$ and $\alpha_i=0 \ \forall i \neq j$, then activating $\alpha_k$ so that $\alpha_k = k+1$ (where $k \neq j$) leads to a decrease in the value of $\mathcal{A}$. Now clearly, if $\alpha_k$ is set equal to $k+1$, the value of $\alpha_j$ has to drop to $p_D-(k+1)$, so that the boundary constraint is not violated. We define the loss incurred by activating $\alpha_k$ by $H$:
\begin{align*}
    H & = \dbinom{p_D}{j+1}\times \frac{1}{(j+1)^2} - \left\{ \dbinom{p_D-k-1}{j+1}\times \frac{1}{(j+1)^2} + \dbinom{k+1}{k+1}\times \frac{1}{(k+1)^2} \right\}\\
    & = \frac{1}{(j+1)^2}\times \left\{\dbinom{p_D}{j+1} - \dbinom{p_D-k-1}{j+1} \right\} - \frac{1}{(k+1)^2}\\ 
    & = \frac{1}{(j+1)^2 \times\Gamma\left(j+2\right)}\left\{ \frac{\Gamma\left(p_D+1\right)}{\Gamma\left(p_D-j\right)} - \frac{\Gamma\left(p_D-k\right)}{\Gamma\left(p_D-k-j-1\right)}\right\} - \frac{1}{(k+1)^2}.
\end{align*}

\noindent We distinguish between two possible cases.\\

\noindent \underline{Case 1:} $j<k$\\

\noindent If $j<k$, we can write $j=k-\beta$, where $\beta \in \mathbb{Z}^+$. Therefore, the loss function $H$ becomes:
\begin{align*}
    H = \frac{1}{(k-\beta+1)^2 \times\Gamma\left(k-\beta+2\right)}\left\{ \frac{\Gamma\left(p_D+1\right)}{\Gamma\left(p_D-k+\beta\right)} - \frac{\Gamma\left(p_D-k\right)}{\Gamma\left(p_D-2k+\beta-1\right)}\right\} - \frac{1}{(k+1)^2}.
\end{align*}

\noindent \underline{Case 2:} $j>k$\\

\noindent If $j>k$, we can write $j=k+\beta$, where $\beta \in \mathbb{Z}^+$. Therefore, the loss function $H$ becomes:
\begin{align*}
    H = \frac{1}{(k+\beta+1)^2 \times\Gamma\left(k+\beta+2\right)}\left\{ \frac{\Gamma\left(p_D+1\right)}{\Gamma\left(p_D-k-\beta\right)} - \frac{\Gamma\left(p_D-k\right)}{\Gamma\left(p_D-2k-\beta-1\right)}\right\} - \frac{1}{(k+1)^2}.
\end{align*}

\noindent The above expressions for $H$ give us additional restrictions, which are that $p_D-2k+\beta-1 > 0$ and $p_D-2k-\beta-1 > 0$, so that all terms are well-defined. These restrict the set of possible values of $k$ and $\beta$, which enables us to perform a simulation study to see how the value of $H$ varies for different combinations of values of $k$ and $\beta$, as well as for increasing $p_D$. In fact, for each value of $p_D$, we are only interested in the combination of $k$ and $\beta$ values that yield the lowest $H$ value, as we seek to find what the minimum loss is and how that behaves for varying $p_D$. Figure \ref{fig:minH_activation} shows the minimum value of $H$ achieved for $9\leq p_D \leq 100$; it can be seen that for both cases of $j<k$ and $j>k$, the minimum value of $H$ is always positive. This implies that no matter what the values of $j$ and $k$ are, activating $\alpha_k$ will always lead to a lower $\mathcal{A}$ value, suggesting that even when we are on the boundary of the solution space, having exactly one parameter equal to $p_D$ should be preferred to having two or more non-zero parameters that sum up to $p_D$.

\begin{figure}[h!]
    \centering
    \includegraphics[scale=0.7]{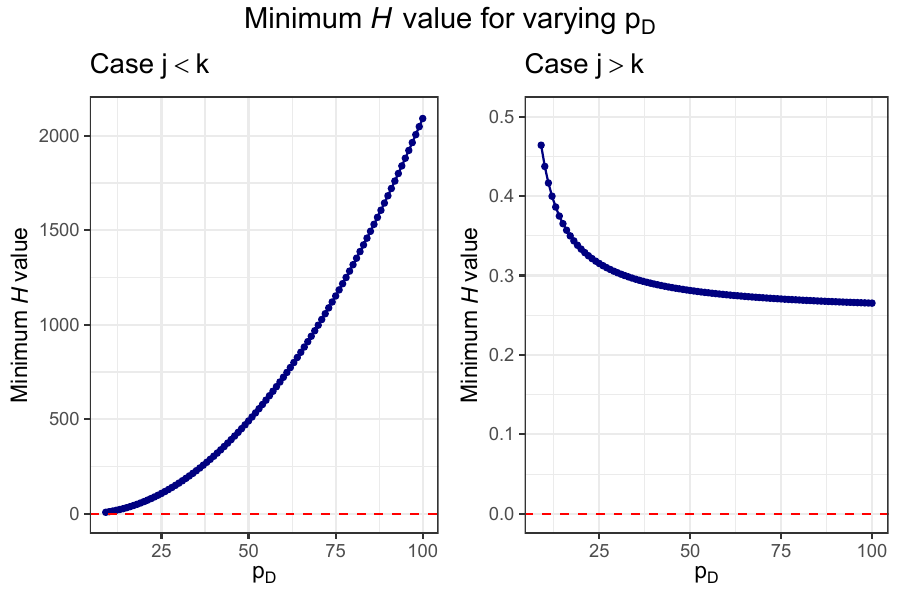}
    \vspace{-0.1cm}
    \caption{Minimum loss ($H$) value caused by the activation of a parameter $\alpha_k$ when $\alpha_j=p_D$, for $9 \leq p_D \leq 100$. The red dashed line at $y=0$ indicates that the loss value is always positive.}
    \label{fig:minH_activation}
\end{figure}

\noindent Thus, we have shown the following:
\begin{enumerate}
    \item The greatest $\mathcal{A}$ value for $p_D \leq 8$ is achieved when all the $\alpha_i$'s are equal to zero.
    \item When $p_D \geq 9$, $\mathcal{A}$ is maximised when we are on the boundary of the solution space.
    \item If $p_D \geq 9$ and we are on the boundary of the solution space with just one parameter $\alpha_j$ being equal to $p_D$, then activating any other parameter $\alpha_k$ (where $j\neq k$) so that it is not longer equal to zero, leads to a lower value of $\mathcal{A}$.
\end{enumerate}
Given the above, we can conclude that the optimal solution for $p_D \leq 8$ is attained when all parameters are equal to zero, while for $p_D \geq 9$, the optimal solution is on the boundary of the solution space with just one parameter $\alpha_i$ being non-zero and thus equal to $p_D$, as required.

\end{proof}
\setcounter{equation}{0}
\renewcommand\theequation{B.\arabic{prop}.\arabic{equation}}
\begin{prop}\label{prop:maxdiscscore}
    The maximum value of the discrete score $s_{D,(i, \cdot)}$ for the $i$th observation of a data set with $p_D$ discrete variables is given by:
    \begin{equation*}
    s^\mathrm{max}_{D,(i,\cdot)} = \left\{
\begin{array}{cl}
      p_D & \mathrm{if}\ 1\leq p_D \leq 8,\\[15pt] 
      \frac{\dbinom{p_D}{\min\left\{\maxlen, 3 \right\}}}{\left(\min\left\{\maxlen, 3 \right\}\right)^2} & \mathrm{if}\ 9\leq p_D \leq 10,\\[15pt]
      \frac{\dbinom{p_D}{\min\left\{\maxlen, \left\lfloor \frac{p_D}{2}-1 \right\rfloor \right\}}}{\left(\min\left\{\maxlen, \left\lfloor \frac{p_D}{2}-1 \right\rfloor \right\}\right)^2} & \mathrm{if}\ p_D \geq 11 \\
\end{array}
\right.
    \end{equation*}
\end{prop}
\begin{proof}
    The most extreme case that may occur is that of all sequences of length $k\leq p_D$ appearing just once (thus having a support of a unit), with infinitesimal association between all possible pairs of discrete variables, and since we make sure $\sigma_d\geq 2$, that would immediately imply they are all infrequent (see Proposition \ref{prop:maxdiscscore_expression} for a justification of this). Moreover, all of them would be pruned and there would be no further contributions to $s_{D,(i,\cdot)}$. This yields a total score of $C^{p_D}_{k}/{k^2}$,where $C$ is the binomial coefficient. We therefore seek to find the sequence length $k^* \in \mathbb{Z}^+$ that maximises this expression of the score, for fixed $p_D$. This means that the score for $k=k^*-1$ and for $k^*+1$ (we assume $1<k^* \leq p_D$) will be less than that for $k^*$, or equivalently:
    \begin{align}
        \dbinom{p_D}{k^*+1}\frac{1}{\left(k^*+1\right)^2}-\dbinom{p_D}{k^*}\frac{1}{\left(k^*\right)^2} & < 0, \label{eq:ineq_maxdiscscore_1}\\
        \dbinom{p_D}{k^*}\frac{1}{\left(k^*\right)^2} - \dbinom{p_D}{k^*-1}\frac{1}{\left(k^*-1\right)^2} & > 0 \label{eq:ineq_maxdiscscore_2}.
    \end{align}
    Starting with Expression \eqref{eq:ineq_maxdiscscore_1}, we have:
    \begin{multline*}
        \dbinom{p_D}{k^*+1}\frac{1}{\left(k^*+1\right)^2}-\dbinom{p_D}{k^*}\frac{1}{\left(k^*\right)^2} =\\
         = \frac{p_D!}{\left( k^* + 1 \right)!\left(p_D - k^* - 1\right)!}\frac{1}{\left(k^*+1 \right)^2} - \frac{p_D!}{\left( k^* \right)!\left(p_D - k^* \right)!}\frac{1}{\left(k^* \right)^2}\\
         = \frac{p_D!}{\left(k^*\right)!\left(p_D-k^*-1\right)!\left(k^*\right)^2}\left\{\frac{1}{\left(k^*+1\right)\left(1+\frac{1}{k^*}\right)^2}-\frac{1}{p_D-k^*} \right\},
    \end{multline*}
    and since we know that the first term is strictly positive, it suffices to show that the expression inside the curly brackets is negative. We bring this to the following form: $$\begin{aligned}
        \frac{1}{\left(k^*+1\right)\left(1+\frac{1}{k^*}\right)^2}-\frac{1}{p_D-k^*} & = \frac{\left(p_D-k^* \right)-\left( k^*+1\right)\left(1+\frac{1}{k^*}\right)^2}{\left(k^*+1 \right)\left(1+\frac{1}{k^*}\right)^2\left(p_D-k^*\right)}\\
    \end{aligned},$$
    hence it suffices to show that the numerator is negative, since the denominator is again strictly positive. This gives:
    \begin{align}
        & \left(p_D-k^* \right)-\left( k^*+1\right)\left(1+\frac{1}{k^*}\right)^2 < 0 \nonumber\\
        \iff & p_D-k^* - \left(k^* + 3 + \frac{3}{k^*}+\frac{1}{\left( k^* \right)^2} \right) < 0 \nonumber\\
        \iff & p_D - 2k^* - 3 - \frac{3}{k^*} - \frac{1}{\left( k^* \right)^2}  < 0 \nonumber\\
        \iff & p_D < 2k^* + 3 + \frac{3}{k^*}+\frac{1}{\left( k^* \right)^2}. \label{eq:ineq1_maxdscore}
    \end{align}
    We proceed similarly to derive an additional bound based on Expression \eqref{eq:ineq_maxdiscscore_2}. More precisely:
    \begin{multline*}
        \dbinom{p_D}{k^*}\frac{1}{\left(k^*\right)^2} - \dbinom{p_D}{k^*-1}\frac{1}{\left(k^*-1\right)^2} = \\
        = \frac{p_D!}{\left( k^* \right)!\left(p_D - k^* \right)!}\frac{1}{\left(k^* \right)^2} - \frac{p_D!}{\left( k^* - 1 \right)!\left(p_D - k^* + 1\right)!}\frac{1}{\left(k^*-1 \right)^2}\\
        = -\frac{p_D!}{\left(k^*-1\right)!\left(p_D-k^*\right)!\left(k^*-1\right)^2}\left\{\frac{1}{p_D-k^*+1} - \frac{\left(1-\frac{1}{k^*} \right)^2}{k^*}\right\},
    \end{multline*}
    so by noticing that the first term is strictly positive and the product should be positive, we require that the expression inside the curly brackets is negative. A bit of re-arrangement gives:
    \begin{align*}
        \frac{1}{p_D-k^*+1} - \frac{\left(1-\frac{1}{k^*} \right)^2}{k^*} =
        \frac{k^*-\left(p_D-k^*+1 \right) \left(1-\frac{1}{k^*} \right)^2}{\left(p_D-k^*+1 \right)k^*},
    \end{align*}
    thus we require that the numerator is negative, since the denominator is strictly positive. Hence, we have:
    \begin{align}
        & k^*-\left(p_D-k^*+1 \right) \left(1-\frac{1}{k^*} \right)^2 < 0 \nonumber\\
        \iff & \frac{k^*}{\left(1-\frac{1}{k^*} \right)^2} < p_D-k^*+1 \nonumber\\
        \iff & \frac{k^*}{\left(1-\frac{1}{k^*} \right)^2} + k^* - 1 < p_D. \label{eq:ineq2_maxdscore}
    \end{align}
    We plot these bounds for $p_D$ in Figure \ref{fig:ineq_region_wzoom2} to get an idea of what the region of solutions looks like.
    \vspace{-.4cm}
    \begin{figure}[h!]
        \centering \centerline{\includegraphics[scale=0.4]{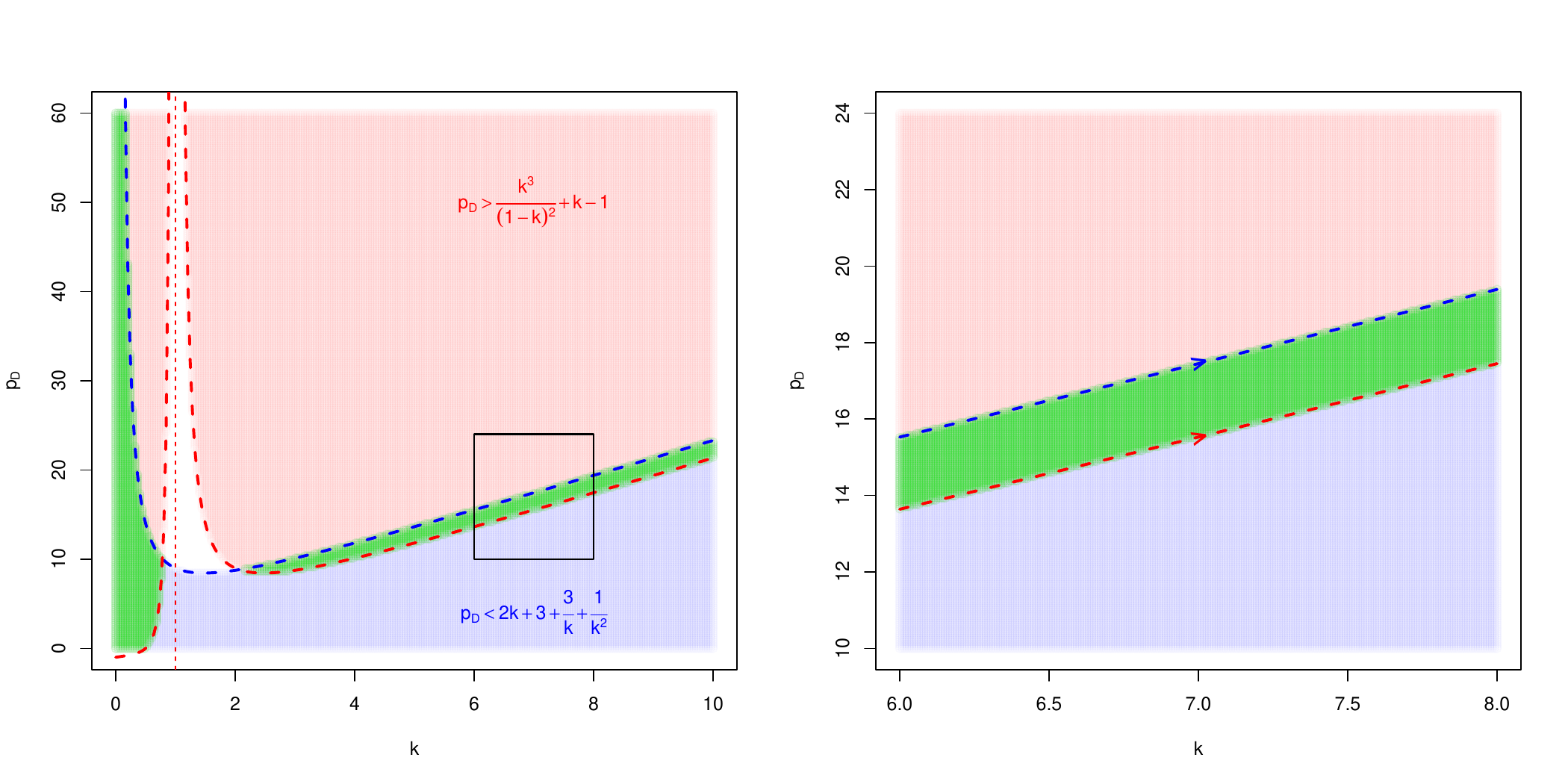}}
        \vspace{-0.1cm}
        \caption{Plot of bounds of $p_D$. The blue shaded region corresponds to the upper bound from Expression \eqref{eq:ineq1_maxdscore} and the red shaded region corresponds to the lower bound from Expression \eqref{eq:ineq2_maxdscore}. The green shaded region is the region where both these bounds are satisfied. The vertical dashed line at $k=1$ is the asymptote of the expression for the lower bound. The rectangular region $[6,8]\times[10,24]$ on the left subplot is zoomed in on the right, illustrating how the 2 boundary curves become parallel as $p_D$ increases.}
        \label{fig:ineq_region_wzoom2}
    \end{figure}
    \noindent We can find the points of intersection of the boundary curves by solving:$$
    \frac{k^3}{\left(1-k\right)^2}+k-1 = 2k+3+\frac{3}{k}+\frac{1}{k^2},$$
    which becomes equivalent to finding the roots of the following polynomial:$$
    -2k^4+4k^3+k^2-k-1 = 0.$$
    The two real roots are $k \approx 2.07$ and $k \approx 0.78$ which correspond to $p_D \approx 8.82$ and $p_D \approx 10.05$, respectively (in 2 decimal places). Given $k^* \in \mathbb{Z}^+$, we need to distinguish between the cases of $p_D \leq 8$ and $p_D \geq 9$. We start with the case of $p_D \geq 9$ and we see that the two boundary curves become parallel as $k \rightarrow +\infty$. We can compute their gradient to validate this:$$\begin{aligned}
        \frac{\mathrm{d}\left(2k+3+\frac{3}{k}+\frac{1}{k^2}\right)}{\mathrm{d}k} = 2-\frac{3}{k^2}-\frac{2}{k^3} & \overset{k\rightarrow \infty}{\longrightarrow} 2, \\
        \frac{\mathrm{d}\left(\frac{k^3}{\left(1-k\right)^2}+k-1\right)}{\mathrm{d}k} = \frac{k^3-3k^2}{k^3-3k^2+3k+1} + 1 & \overset{k\rightarrow \infty}{\longrightarrow} 2.
    \end{aligned}$$
    Therefore, we can access the set of solutions (corresponding to the region where both bounds are satisfied) by considering the line $p_D = 2k+\alpha$. The constant $\alpha$ can be computed by noticing that the aforementioned line will pass through the midpoint of each line segment defined by the two boundary curves for large $k$. Hence, we need to compute the difference between the two boundary curves:$$\begin{aligned}
        2k+3+\frac{3}{k}+\frac{1}{k^2} - \frac{k^3}{\left(1-k\right)^2}-k+1 = \frac{2k^4-4k^3-k^2+k+1}{k^2\left(k-1\right)^2} \overset{k\rightarrow \infty}{\longrightarrow} 2,
    \end{aligned}$$
    which means that for large $k$, the vertical distance between the two boundary curves is equal to two units. Their midpoint is just a unit away from each curve, thus we require:$$\begin{aligned}
        2k+3+\frac{3}{k}+\frac{1}{k^2} - (2k + \alpha) = 3 - \alpha + \frac{3}{k} + \frac{1}{k^2} \overset{k\rightarrow \infty}{\longrightarrow} 1,\\
    \end{aligned}$$
    which is satisfied for $\alpha=2$. This means that for large $k$ values, we can compute $k^*$ by accessing the region where the bounds for $p_D$ are satisfied via the line $p_D = 2k+2$. We need to find the values of $k$ (and thus the values of $p_D$) for which this can be used. We can only use this line to determine $k^*$ as long as it outputs a value greater than the lower bound for $p_D$, hence we need:$$\begin{aligned}
        & 2k+2 > \frac{k^3}{\left(1-k\right)^2}+k-1 \\
        \iff & k^2 - 5k + 3 > 0,
    \end{aligned}
    $$
    which holds for the following set of values of $k$:$$
     k \in \left(-\infty, \frac{-\sqrt{13}+5}{2}\right) \cup \left( \frac{\sqrt{13}+5}{2}, +\infty \right) \approx (-\infty, -0.70) \cup (4.30, +\infty).$$
     Recalling that $k$ is non-negative, the interval containing negative values can safely be rejected and we get $k\approx 4.30$, which corresponds to $p_D \approx 10.61$, meaning that we can only use the line $p_D = 2k+2$ for $p_D \geq 11$. Moreover, notice that using this line gives $k = p_D/2-1$, which means that for odd $p_D$ we do not get an integer value for $k$. However, we can see that for $p_D=11$, the only integer $k^*$ such that the bounds for $p_D$ are satisfied is $k^*=4$, while for $p_D = 12$ and $p_D = 13$, the only such integer is $k^*=5$. Since the gradient of the line is equal to 2, there can only be one integer value of $k$ satisfying the bounds for every two consecutive $p_D$ values, which we can calculate using:$$
     k^* = \left\lfloor\frac{p_D}{2}-1\right\rfloor,$$
     where $\lfloor \cdot \rfloor$ is the floor function ($\lfloor x \rfloor$ is the greatest integer $x'$ such that $x' \leq x$).\\

     \noindent For $p_D = 9$ and $p_D = 10$, $k^*=3$ is the only integer satisfying the derived bounds. Then, for $p_D \leq 8$, there is no integer satisfying Expressions \eqref{eq:ineq_maxdiscscore_1} \& \eqref{eq:ineq_maxdiscscore_2} and this is because there does not exist any $k^* \in \mathbb{Z}^+$ such that Expression \eqref{eq:ineq_maxdiscscore_2} holds (this can also be seen from Figure \ref{fig:ineq_region_wzoom2}). This implies that the expression for the score is decreasing for $k$. Hence, its maximum is achieved for the lowest possible integer value of $k$ which is equal to one. Summarising the above, the sequence length that maximises the discrete score $s_{D,(i,\cdot)}$ is given by:
     \begin{equation*}
    k^* = \left\{
\begin{array}{cl}
      1 & \mathrm{if}\ 1\leq p_D \leq 8,\\ 
      3 & \mathrm{if}\ 9\leq p_D \leq 10,\\
      \left\lfloor\frac{p_D}{2}-1\right\rfloor & \mathrm{if}\ p_D \geq 11 \\
\end{array}
\right.
    \end{equation*}
    The proof is completed by considering that $\maxlen$ is always at least equal to one; thus the maximum score for $1 \leq p_D \leq 8$ is equal to $p_D$, while for $p_D \geq 9$, the maximum score is equal to the expression that we mentioned at the beginning of the proof, evaluated at $k=k^*$. However, since $\maxlen$ can be less than $k^*$ and given that the maximum score expression is increasing for $k \in \mathbb{Z}^+_{\leq k^*}$, the maximum discrete score is attained for $k = \min\{k^*, \maxlen\}$.
\end{proof}

\setcounter{equation}{0}
\renewcommand\theequation{B.\arabic{prop}.\arabic{equation}}
\begin{prop}\label{prop:maxcontrib}
The maximum value of the contribution $c_{D,(i,j)}$ of the $i$th observation to the $j$th discrete variable of a data set with $p_D$ discrete variables is given by:
\footnotesize
\begin{equation*}
    c^\mathrm{max}_{D,(i,j)} = \left\{
\begin{array}{cl}
      p_D & \mathrm{if}\ 1\leq p_D \leq 13,\\[15pt] 
      \frac{\dbinom{p_D}{\min\left\{\maxlen, 5 \right\}\times \mathbbm{1}\left\{\maxlen\geq 4\right\} + \mathbbm{1}\left\{\maxlen\leq 3\right\}}}{\left(\min\left\{\maxlen, 5 \right\}\right)^3\times \mathbbm{1}\left\{\maxlen\geq 4\right\} + \mathbbm{1}\left\{\maxlen\leq 3\right\}} & \mathrm{if}\ p_D=14,\\[15pt]
      \frac{\dbinom{p_D}{\min\left\{\maxlen, 5 \right\}\times \mathbbm{1}\left\{\maxlen\geq 3\right\} + \mathbbm{1}\left\{\maxlen\leq 2\right\}}}{\left(\min\left\{\maxlen, 5 \right\}\right)^3\times \mathbbm{1}\left\{\maxlen\geq 3\right\} + \mathbbm{1}\left\{\maxlen\leq 2\right\}} & \mathrm{if}\ p_D=15,\\[15pt]
      \frac{\dbinom{p_D}{\min\left\{\maxlen, 6 \right\}\times \mathbbm{1}\left\{\maxlen\geq 3\right\} + \mathbbm{1}\left\{\maxlen\leq 2\right\}}}{\left(\min\left\{\maxlen, 6 \right\}\right)^3\times \mathbbm{1}\left\{\maxlen\geq 3\right\} + \mathbbm{1}\left\{\maxlen\leq 2\right\}} & \mathrm{if}\ p_D=16,\\[15pt]
      \frac{\dbinom{p_D}{\min\left\{\maxlen, 6 \right\}}}{\left(\min\left\{\maxlen, 6 \right\}\right)^3} & \mathrm{if}\ p_D=17,\\[15pt]
      \frac{\dbinom{p_D}{\min\left\{\maxlen, \left\lfloor \frac{p_D}{2}-\frac{5}{4} \right\rfloor \right\}}}{\left(\min\left\{\maxlen, \left\lfloor \frac{p_D}{2}-\frac{5}{4} \right\rfloor \right\}\right)^3} & \mathrm{if}\ p_D \geq 18 \\
\end{array}
\right.
    \end{equation*}
%\end{adjustwidth}
\normalsize
\end{prop}
\begin{proof}
    The proof is the same as that of Proposition \ref{prop:maxdiscscore}, except we now seek to maximise the expression $C^{p_D}_k/k^3$ for $k \in \mathbb{Z}^+_{\leq p_D}$. For $p_D \leq 13$, there does not exist any integer values of $k$ such that $C^{p_D}_k/k^3 > p_D$, therefore the maximum contribution is equal to $p_D$ and it can be attained if all $p_D$ categorical levels of an observation are unique within the discrete variable they belong to.\\
    
    \noindent Using the same strategy as in Proposition \ref{prop:maxdiscscore} gives the expression for $p_D \geq 18$. The cases $p_D \in \{14, 15, 16, 17\}$ need to be investigated individually due to the monotonicity of the expression $C^{p_D}_k/k^3$ being different to that for lower or greater values of $p_D$. More precisely, the function presents local minima between $k=1$ and the maximising value of $k$, thus imposing additional constraints on the maximum score that depend on the value of $\maxlen$, which are incorporated in the final expression for $c^\mathrm{max}_{D,(i,j)}$.
\end{proof}
\newpage

\setcounter{equation}{0}
\renewcommand\theequation{B.\arabic{prop}.\arabic{equation}}
\begin{prop}\label{prop:minpvalue}
Let a discrete variable have $\ell$ levels and suppose we wish to test the following hypothesis:
\begin{align*}
    H_0: \ \hat{\boldsymbol{\pi}}_l = \boldsymbol{\pi} \ \text{vs. } H_1: \ \hat{\boldsymbol{\pi}}_l \neq \boldsymbol{\pi},
\end{align*}
where $\boldsymbol{\pi}$ is the vector of proportions of the $\ell$ levels in the data set and $\boldsymbol{\hat{\pi}}_l$ is the vector of proportions of the $\ell$ levels in the $k$ nearest neighbours of the core point of level $l$ ($l=1, \ldots, \ell$) with respect to a suitable distance metric and for a given set of continuous variables. Assuming that an association exists, $H_0$ is rejected for all $\ell$ levels at a pre-specified significance level (upon correction) and we further assume that the $l$th element of $\hat{\boldsymbol{\pi}}_l$ is at least equal to $1/2$. The smallest $p$-value for such a chi-squared goodness-of-fit test is obtained when all $k$ nearest neighbours of the core point of level $l$ are also of the same level.
\end{prop}
\begin{proof}
Using the assumption that $\hat{\boldsymbol{\pi}}_l \geq 1/2$, it follows that the observed points of class $l$ are at least given by $k/2$, where $k$ is the number of closest neighbours we consider. The chi-squared test statistic for this test (which we denote by $\scalebox{1.5}{$\chi$}^2_l$ in a slight abuse of notation), is then given by:
\begin{equation}\label{eq:chisquaredteststat}
    \scalebox{1.5}{$\chi$}^2_l= \sum\limits_{i=1}^\ell \frac{\left(E_{i,l} - O_{i,l} \right)^2}{E_{i,l}},
\end{equation}
where $E_{i,l}$ is the expected amount of neighbours of the $l$th core point which are of level $i$ and $O_{i,l}$ is the amount of them that is observed. We assume $k < \min_i n_i$ (where $n_i$ is the number of observations possessing level $i$) and we know that $E_{i,l} = k\pi_i$, $O_{l,l} \geq k/2$ and $\sum_i O_{i,l} = k$, which imply that $\sum_{i \neq l}O_{i,l} \leq k/2$. Assume w.l.o.g. that $l=1$ and observe that $O_{\ell,1} = k - \sum_{i=1}^{\ell-1} O_{i, 1}$ and $\pi_\ell = 1- \sum_{i=1}^{\ell-1}\pi_i$. Expression \eqref{eq:chisquaredteststat} then becomes:
\begin{equation}\label{eq:chisquaredteststat2}
    \scalebox{1.5}{$\chi$}^2_1= \sum\limits_{i=1}^{\ell-1} \frac{\left(E_{i,1} - O_{i,1} \right)^2}{E_{i,1}} + \frac{\left\{k\left(1-\sum\limits_{i=1}^{\ell-1}\pi_i \right) - \left( k - \sum\limits_{i=1}^{\ell-1}O_{i,1} \right)\right\}^2}{k\left(1-\sum\limits_{i=1}^{\ell-1}\pi_i \right)}
\end{equation}
Taking the derivative of Expression \eqref{eq:chisquaredteststat2} with respect to $O_{1,1}$ (these are integers but we extend the function to the reals and show that the expression is maximised for integer-valued $O_{1,1}$) we obtain the following expression:
\begin{equation}\label{eq:gradchisquared}
    \frac{\partial\scalebox{1.5}{$\chi$}^2_1}{\partial O_{1,1}} = \frac{2O_{1,1}}{k\pi_1} + \frac{2\left( k - \sum\limits_{i=1}^{\ell-1} O_{i,1} \right)}{k\left(1-\sum\limits_{i=1}^{\ell-1}\pi_i \right)}.
\end{equation}
Expression \eqref{eq:gradchisquared} consists of a strictly positive term, while the second term can only be non-negative if $\sum_{i=1}^{\ell-1}O_{i,1} = k$, meaning that $O_{\ell,1} = 0$. We then proceed in the exact same way, re-writing Expression \eqref{eq:chisquaredteststat2}, substituting $\sum_{i=1}^{\ell-1}O_{i,1}$ by $k$, taking the derivative with respect to $O_{i,1}$ and then we observe that $\sum_{i=1}^{\ell-2}O_{i,1} = k$. This recursive argument leads to $O_{1,1} = k$ and $O_{i,1} = 0 \ \forall i > 1$. Geometrically, Expression \eqref{eq:chisquaredteststat2} corresponds to an elliptic paraboloid when plotted on the set of coordinate axes defined by $\left(O_{1,1}, \ldots, O_{\ell-1, 1}\right)$. The cross section that attains a maximum is observed for $O_{1,1} = k$ and the maximum point is attained when $O_{i,1} = 0 \ \forall i > 1$. This is depicted in Figure \ref{fig:surface_plt}, where $\scalebox{1.5}{$\chi$}^2_1$ is plotted in the $\left(O_{1,1}, O_{2,1}\right)$ space for $k=10$, $\ell = 3$ and $\boldsymbol{\pi} = (0.4, 0.35, 0.25)$.\\

\noindent This result justifies the use of the product of $p$-values as a criterion for deciding the set of continuous variables for which the levels are best separated. If there exists a structure in the data and assuming that can be captured by the closest neighbours of the most centrally located point of a level with respect to some distance metric, then the null hypothesis should be comfortably rejected. We have shown that the amount of `deviation' from the null distribution is proportional to the value of the test statistic, with maximal test statistic value obtained when the $k$ neighbours of the core point of a level $l$ are all of the same level. Therefore, a lower product of $p$-values will correspond to a larger test statistic which indicates a larger deviation from the null distribution and thus a better separation of the $\ell$ levels. Finally, taking products of $p$-values ensures a fair comparison, as imbalanced levels produce $p$-values on different scales (less frequent levels yield lower $p$-values); this is accounted for by comparing products of all $\ell$ $p$-values for different subsets of continuous features.
\vspace{-.4cm}
\begin{figure}[h!]
  \centering
  \includegraphics[scale=0.85]{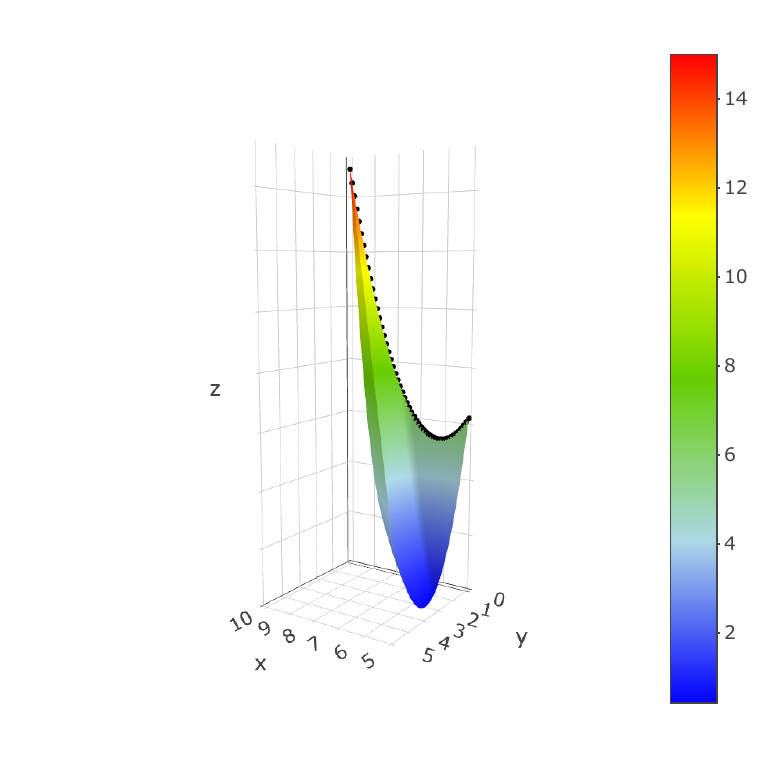}
  \caption{Plot of $\scalebox{1.5}{$\chi$}^2_1$ in the $\left(O_{1,1}, O_{2,1}\right)$ space for $\ell = 3$, $k=10$ and $\boldsymbol{\pi} = (0.4, 0.35, 0.25)$; the $x$ and $y$ axes correspond to $O_{1,1}$ and $O_{2,1}$ respectively, with $z = \scalebox{1.5}{$\chi$}^2_1$. The black dots display the cross section (parabola) for which the maximum $\scalebox{1.5}{$\chi$}^2_1$ value is achieved, which is for $O_{2,1} = 0$.}
  \label{fig:surface_plt}
\end{figure}

\end{proof}

\newpage 
\section{Algorithms}\label{appendix:algorithms}
\setcounter{algorithm}{0}
\renewcommand\thealgorithm{C.\arabic{algorithm}}
\begin{algorithm}[h!]
  \caption{Detection of marginal outliers in discrete space}
  \label{code:marg_cat_outs}
  \begin{algorithmic}[1]
  \small
\State \textbf{Input:} Computed discrete scores $s_{D,(i,\cdot)} \ (i=1,\dots,n)$, max proportion of outliers believed to be in the data set $\rho$, additional proportion of outliers we are willing to tolerate $\epsilon$.
\State Find the number of unique discrete score values $\lvert \mathcal{E}\rvert$.
\For{$K=1,2,\dots, \lvert \mathcal{E}\rvert$} 
    \State Run 1-dimensional $K$-Means clustering on the unique score values.
    \State Find $\mathcal{C}^K_0$ that includes the zero score and store its size $\lvert\mathcal{C}^K_0 \rvert$.
\EndFor
\State Set $K^* \leftarrow \min_{K} \{K: \lvert \mathcal{C}_0^K \rvert = M \}$, where $M$ is the mode of $\{\lvert \mathcal{C}_0^K \rvert\}_{K=1}^{\lvert\mathcal{E}\rvert}$.
\State Scale the scores using $s'_{D,(i,\cdot)}=\lvert s_{D,(i,\cdot)}-m_D \rvert/s_D, \ i=1,\dots,\lvert \mathcal{E}\rvert$, where $m_D$ is the mean discrete score and $s_D$ the standard error of the discrete scores.
\State Sort the scaled scores in ascending order and find the lowest index $j$ for which $s'_{D,(j+1,\cdot)} > s'_{D,(j,\cdot)} + 1$.
\State Remove observations with infrequent sequences of unit length or with score greater than $s'_{D,(j,\cdot)}$ from $\mathcal{C}_0^{K^*}$.
\If{$\lvert \mathcal{C}_0^{K^*} \rvert \leq n- \lceil (\rho+\epsilon)n \rceil$}
    \State Set $K^* \leftarrow K^*-1$ and go back to Step 10.
\EndIf
\State \textbf{Output:} Set of marginal outliers in the discrete space $\mathcal{I}_D$, i.e. observations not in $\mathcal{C}_0^{K^*}$.
\end{algorithmic}
\end{algorithm}

\begin{algorithm}[h!]
  \caption{Detection of marginal outliers in continuous space}
  \label{code:marg_cont_outs}
  \begin{algorithmic}[1]
  \small
\State \textbf{Input:} Computed continuous scores $s_{C,(i,\cdot)}, \ i=1,\dots,n$, set of indices of marginal outliers in discrete space $\mathcal{I}_D$, max proportion of outliers believed to be in the data set $\rho$, additional proportion of outliers we are willing to tolerate $\epsilon$.
\State Discard observations in $\mathcal{I}_D$.
\State Sort the continuous scores for remaining observations in ascending order and compute differences among consecutive sorted continuous scores $z_i$.
\State Scale the differences using $z'_i=\lvert z_i-m_C \rvert/s_C$, where $m_C$ and $s_C$ are the mean and the standard error of the $z_i$'s, respectively.
\For{$\lambda=2,\dots, 20$} 
    \State Calculate $\lvert \Delta_\lambda \rvert$, where $\Delta_\lambda = \{ z'_i: z'_i \geq \lambda, i \in \{1,\dots, n\}\backslash \mathcal{I}_D\}$.
\EndFor
\State Set $\lambda^* \leftarrow \max_{\lambda} \{\lambda: \lvert \Delta_\lambda \rvert = M \}$, where $M$ is the mode of $\{\lvert \Delta_\lambda \rvert \}_{\lambda=2}^{20}$.
\If{$\lvert \Delta_{\lambda^*} \rvert + \lvert \mathcal{I}_D \rvert \geq \lceil (\rho+ \epsilon)n \rceil$}
    \State \parbox[t]{\dimexpr\textwidth-\leftmargin-\labelsep-\labelwidth}{Go back to Step 10 of Algorithm \ref{code:marg_cat_outs} and set $K^* \leftarrow K^*-1$ to get a new set of marginal outliers in the discrete space 
    $\mathcal{I}_D$.\strut}
\EndIf
\State \textbf{Output:} Set of marginal outliers in the continuous space $\mathcal{I}_C$, i.e. observations with continuous score greater than the lowest score that differs at least by $\lambda^*$ from its next higher score.
\end{algorithmic}
\end{algorithm}

\begin{algorithm}[h!]
  \renewcommand{\break}{\textbf{break}}
  \caption{Context identification for a given discrete behavioural attribute for one dimension}
  \label{code:context_id_1d}
  \begin{algorithmic}[1]
  \small
\State \textbf{Input:} Data set $\boldsymbol{X}$, set of indices of marginal outliers $\mathcal{I}$, index of discrete variable $j$, proportion parameter $\delta$, significance levels $\alpha_1, \alpha_2$, order of Minkowski distance $r$.
\State Discard marginal outliers $\mathcal{I}$ from $\boldsymbol{X}$.
\State Define empty list $\boldsymbol{Z}_1$.
\For{$k=1, \ldots, p_C$}
    \State Perform the Kruskal-Wallis $H$ test on $\left( \boldsymbol{X}_{D_j}, \boldsymbol{X}_{C_k}\right)$ to get a $p$-value $p_k$.
    \If{$p_k \leq \alpha_1$}
        \State Store $k$ in $\mathcal{J}$.
    \EndIf
\EndFor
\For{$k \in \mathcal{J}$}
    \State \parbox[t]{\dimexpr\textwidth-\leftmargin-\labelsep-\labelwidth}{Construct a distance matrix $\boldsymbol{\mathrm{A}}$ for $\boldsymbol{X}_{C_k}$ with the Minkowski distance of order $p$.}
    \For{$l = 1, \ldots , \ell_j$}
        \State \parbox[t]{\dimexpr\textwidth-\leftmargin-\labelsep-\labelwidth}{Extract the distance matrix $\boldsymbol{\mathrm{A}}^{l}$ from $\boldsymbol{\mathrm{A}}$ for points for which $\boldsymbol{X}_{D_j} = l$.}
        \State \parbox[t]{\dimexpr\textwidth-\leftmargin-\labelsep-\labelwidth}{Find the most central point of level $l$ as $\boldsymbol{x}^*_l \defeq \argmin\limits_{y \in \{1, \ldots, n_l\}} \text{Median}\left(\boldsymbol{\mathrm{A}}^{l}_{y} \right)$.\strut}
        \State \parbox[t]{\dimexpr\textwidth-\leftmargin-\labelsep-\labelwidth}{Compute the $\lceil \delta \times n_l \rceil$-nearest neighbours of $\boldsymbol{x}^*_l$ and store the number of points from each level in $\Tilde{\boldsymbol{\pi}}_j$.\strut}
        \State \parbox[t]{\dimexpr\textwidth-\leftmargin-\labelsep-\labelwidth}{Compute the $p$-value $p_l$ for a chi-square goodness of fit test between $\boldsymbol{\pi}_j$ and $\Tilde{\boldsymbol{\pi}}_j$ and store it in a vector $\boldsymbol{P}_k$.\strut}
    \EndFor
    \State Order $\boldsymbol{P}_k$ in increasing order where $\boldsymbol{P}_k^{(i)}$ is the $i$th smallest $p$-value.
    \If{$\boldsymbol{P}_k^{(i)} \leq \alpha_2/(\ell_j+1-i) \forall i \in \{1, \ldots, \ell_j\}$}
        \State Store $\sum_{i=1}^{\ell_j}\log \left(\boldsymbol{P}_k^{(i)}\right)$ in $\boldsymbol{Z}_1$.
    \EndIf
\EndFor
\If{$\lvert \boldsymbol{Z}_1 \rvert \geq 1$}
    \State Set $\mathcal{J} \leftarrow \mathcal{J}_{z^*}$, where $z^* = \argmin\limits_{z \in \{ 1, \ldots, \lvert \boldsymbol{Z}_1 \rvert \}} \boldsymbol{Z}_1$.
\EndIf
\State \textbf{Output:} Set of contextual attributes $\mathcal{J}$.
\end{algorithmic}
\end{algorithm}

\begin{algorithm}[h!]
  \renewcommand{\break}{\textbf{break}}
  \caption{Context identification for a given discrete behavioural attribute for $>1$ dimensions}
  \label{code:context_id_hd}
  \begin{algorithmic}[1]
  \small
\State \textbf{Input:} Index of discrete variable $j$, contextual attributes $\mathcal{J}$, proportion parameter $\delta$, significance level $\alpha_2$, order of Minkowski distance $r$.
\State Set $\xi \leftarrow \lvert \mathcal{J}\rvert$.
\State Define empty lists $\boldsymbol{Z}_1$ and $\boldsymbol{Z}_2$.
\State Set $\mathcal{B} \leftarrow \mathcal{J}$.
\While{$\xi \leq 2$}
    \State Define empty lists $\boldsymbol{Q}_1$ and $\boldsymbol{Q}_2$.
    \For{$Q \in \{\mathcal{Q} \subseteq \mathcal{J} : \lvert \mathcal{Q} \rvert = \xi \wedge \mathcal{Q} \subseteq \mathcal{B}\}$}
        \For{$l=1, \ldots, \ell_j$}
            \State Define $\boldsymbol{X}^l_{C_Q}$ as the set of observations from $\boldsymbol{X}_{C_Q}$ for which $\boldsymbol{X}_{D_j} = l$.
            \State \parbox[t]{\dimexpr\textwidth-\leftmargin-\labelsep-\labelwidth}{Perform ROBPCA on $\boldsymbol{X}^l_{C_Q}$ to obtain a matrix of loadings $\boldsymbol{\mathrm{V}}^l_Q$ and a corresponding set of eigenvalues $\left(\lambda^l_1, \ldots, \lambda^l_\xi \right)$.\strut}
            \State Project the data on the rotated space so that $\boldsymbol{X}^{*l}_{C_Q} \leftarrow \boldsymbol{X}^l_{C_Q}\boldsymbol{\mathrm{V}}^l_Q$.
            \State \parbox[t]{\dimexpr\textwidth-\leftmargin-\labelsep-\labelwidth}{Compute a distance matrix $\boldsymbol{\mathrm{A}}$ by scaling the Minkowski distance of order $p$ using weights $\left( 1/\lambda_1^l \ldots, 1/\lambda_\xi^l\right)$.\strut}
            \State Extract the distance matrix $\boldsymbol{\mathrm{A}}^{l}$ from $\boldsymbol{\mathrm{A}}$ for points for which $\boldsymbol{X}_{D_j} = l$.
            \State \parbox[t]{\dimexpr\textwidth-\leftmargin-\labelsep-\labelwidth}{Find the most central point of level $l$ as $\boldsymbol{x}^*_l \defeq \argmin\limits_{y \in \{1, \ldots, n_l\}} \text{Median}\left(\boldsymbol{\mathrm{A}}^{l}_{y} \right)$.}
            \State \parbox[t]{\dimexpr\textwidth-\leftmargin-\labelsep-\labelwidth}{Compute the $\lceil \delta \times n_l \rceil$-nearest neighbours of $\boldsymbol{x}^*_l$ and store the number of points from each level in $\Tilde{\boldsymbol{\pi}}_j$.\strut}
            \State \parbox[t]{\dimexpr\textwidth-\leftmargin-\labelsep-\labelwidth}{Compute the $p$-value $p_l$ for a chi-square goodness of fit test between $\boldsymbol{\pi}_j$ and $\Tilde{\boldsymbol{\pi}}_j$ and store it in a vector $\boldsymbol{P}_{Q}$.\strut}
        \EndFor
        \State Order $\boldsymbol{P}_Q$ in increasing order where $\boldsymbol{P}_Q^{(i)}$ is the $i$th smallest $p$-value.
        \If{$\boldsymbol{P}_Q^{(i)} \leq \alpha_2/(\ell_j+1-i) \forall i \in \{1, \ldots, \ell_j\}$}
            \State Store $\sum_{i=1}^{\ell_j}\log \left(\boldsymbol{P}_Q^{(i)}\right)$ in $\boldsymbol{Q}_1$ and $Q$ in $\boldsymbol{Q}_2$.
        \EndIf
    \EndFor
    \If{$\lvert \boldsymbol{Q}_1 \rvert > 0$}
        \State Store $z^* \defeq \argmin\limits_{z \in \{1, \ldots, \lvert \boldsymbol{Q}_1\rvert\}}\boldsymbol{Q}_1$ in $\boldsymbol{Z}_1$.
        \State Set $\mathcal{B} \leftarrow \boldsymbol{Q}_{2_{z^*}}$ and store $\mathcal{B}$ in $\boldsymbol{Z}_2$.
        \State Set $\xi \leftarrow \xi-1$.
    \Else
        \State \break
    \EndIf
\EndWhile
\If{$\lvert \boldsymbol{Z}_1 \rvert \geq 1$}
    \State Set $\mathcal{J} \leftarrow \boldsymbol{Z}_{2_{\Tilde{z}}}$, where $\Tilde{z} \defeq \argmin\limits_{z \in \{ 1, \ldots , \lvert \boldsymbol{Z}_1 \rvert \}} \boldsymbol{Z}_1$.
\EndIf
\State \textbf{Output:} Set of contextual attributes $\mathcal{J}$.
\end{algorithmic}
\end{algorithm}
\clearpage
\begin{algorithm}[ht]
  \caption{Method of consecutive angles}
  \label{code:consecutive_angles}
  \begin{algorithmic}[1]
  \small
\State \textbf{Input:} Amount of density estimation misclassifications $N(1)$, KDE ratio threshold $\Lambda^*_m$ if method returns too many misclassifications, parameter controlling tolerance level for drop in number of misclassifications $\gamma$.
\For{$\Lambda^*=1.5, 2, \dots, 19.5, 20$} 
    \State Compute $\theta_{\Lambda^*} \leftarrow \arctan\{ 2[N(\Lambda^*-0.5)-N(\Lambda^*)]\}$.
    \If{$\Lambda^* \geq 2$ \textbf{AND} $\theta_{\Lambda^*} = \theta_{\Lambda^*-0.5}$ \textbf{AND} $N\left(\Lambda^*-0.5\right)-N\left(\Lambda^*\right) < \gamma$}
    \State Set $\Lambda^* \leftarrow \Lambda^*-0.5$.
    \State \textbf{break}
    \EndIf
\EndFor
\If{$\Lambda^*>11$}
    \State Set $\Lambda^* \leftarrow \Lambda^*_m$.
\EndIf
\State \textbf{Output:} Threshold value $\Lambda^*$.
\end{algorithmic}
\end{algorithm}
\clearpage

\section{Additional Figures \& Tables}\label{appendix:extrafigstables}
\setcounter{figure}{0}
\renewcommand\thefigure{D.\arabic{figure}}
\setcounter{table}{0}
\renewcommand\thetable{D.\arabic{table}}

\begin{figure}[h!]
  \centering
  \includegraphics[scale=0.8]{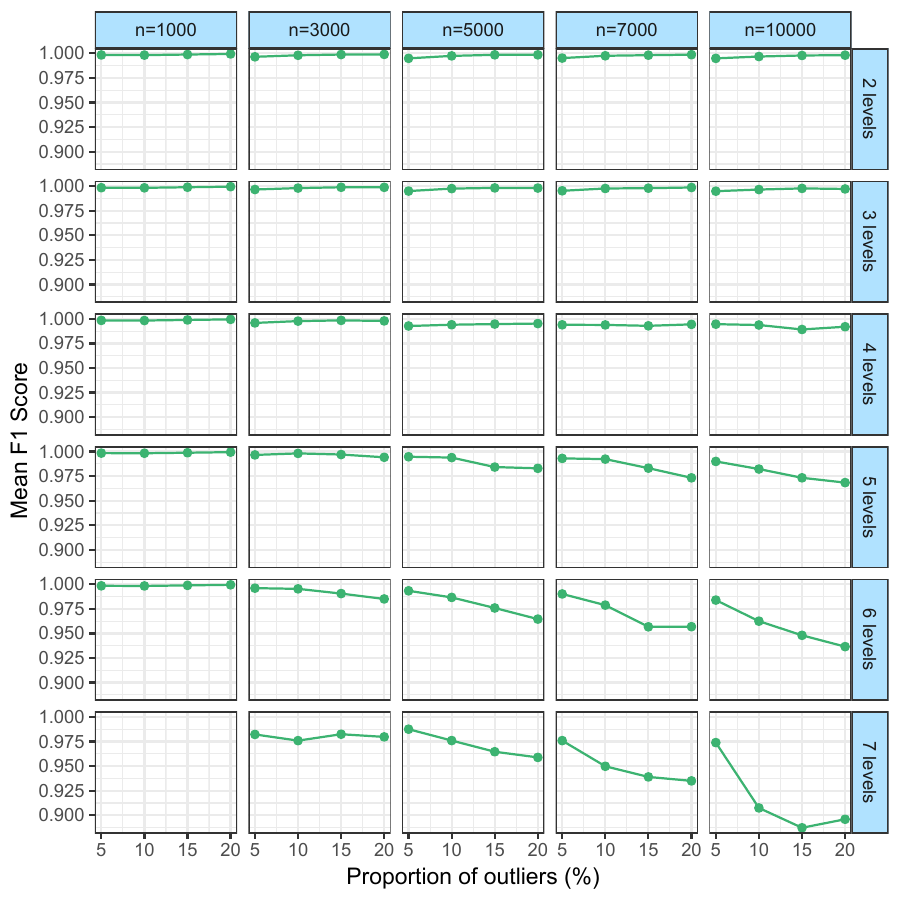}
  \caption{Mean F1 score for simulation study on detection of marginal outliers in artificial data sets for varying number of observations and discrete levels, as well as for increasing proportion of outliers.}
  \label{fig:res_marg_f1}
\end{figure}
\vspace{-1cm}
\clearpage
\begin{table}[ht]
\centering
\begin{tabular}{c|c|c}
Number of continuous variables & Number of levels & $\theta_\mathrm{thresh}(^\circ)$ \\ \hline
\multirow{5}{*}{3} & 3 & 167.50 \\
 & 4 & 168.00 \\
 & 5 & 168.10 \\
 & 6 & 180.00 \\
 & 7 & 180.00 \\ \hline
\multirow{5}{*}{4} & 3 & 166.60 \\
 & 4 & 167.90 \\
 & 5 & 168.00 \\
 & 6 & 180.00 \\
 & 7 & 180.00
\end{tabular}
\caption{Optimal $\theta_\mathrm{thresh}$ values for relationships between 3 \& 4 continuous and one discrete variables with 3--7 levels; $\theta_\mathrm{thresh}$ is the threshold value for $\theta_\mathrm{elbow}$, determining whether a small $\Lambda^*$ value of 2 or 3 or the method of consecutive angles should be used. The binary case is omitted as it is usually recommended to choose $\Lambda^*=2$ or 3 in that case.}
\label{tab:thetathresh34cont}
\end{table}

\end{appendices}

%\renewpagestyle{plain}{%
%\sethead[\thepage][][\textit{REFERENCES}]{}{}{\thepage}
%\setfoot[\thepage][1][1]{}{}{\thepage}
%}%
%\pagestyle{plain}
\end{document}